\documentclass[a4paper,11pt]{article}
\usepackage[twoside,a4paper,margin=1in]{geometry}
\usepackage[utf8]{inputenc}
\usepackage{microtype,graphicx,url,amssymb,bbm,bm,hyperref,color,framed,mathrsfs,comment,amsmath,amsthm}
\usepackage{subcaption}
\usepackage{enumerate}
\usepackage{xcolor}
\usepackage{tcolorbox}
\usepackage{thm-restate}
\usepackage{xspace}
\usepackage{amscd}
\usepackage{float}
\usepackage{graphics}

\graphicspath{{pictures/}}
\usepackage{latexsym}
\usepackage{calc}

\theoremstyle{plain}
\newtheorem{theorem}{Theorem}[section]

\newtheorem{conjecture}[theorem]{Conjecture}
\newtheorem{corollary}[theorem]{Corollary}

\newtheorem{lemma}[theorem]{Lemma}
\newtheorem{remark}[theorem]{Remark}

\usepackage[font={small,it}]{caption}

\title{Short Topological Decompositions of Non-Orientable Surfaces\thanks{This work was partially supported by the ANR project SoS (ANR-17-CE40-0033).}} 
\author{Niloufar Fuladi\thanks{Univ Gustave Eiffel, CNRS, LIGM, F-77454 Marne-la-Vallée, France} \and Alfredo Hubard\footnotemark[2] \and Arnaud de Mesmay\footnotemark[2]}

\begin{document}

\maketitle
\begin{abstract}
In this article, we investigate short topological decompositions of non-orientable surfaces and provide algorithms to compute them. Our main result is a polynomial-time algorithm that for any graph embedded in a non-orientable surface computes a canonical non-orientable system of loops so that any loop from the canonical system intersects any edge of the graph in at most 30 points. The existence of such short canonical systems of loops was well known in the orientable case and an open problem in the non-orientable case. Our proof techniques combine recent work of Schaefer-\v{S}tefankovi\v{c} with ideas coming from computational biology, specifically from the signed reversal distance algorithm of Hannenhalli-Pevzner. 
The existence of short canonical non-orientable systems of loops confirms a special case of a conjecture of Negami on the joint crossing number of two embeddable graphs. We also provide a correction for an argument of Negami bounding the joint crossing number of two non-orientable graph embeddings.
\end{abstract}

\section{Introduction}

\subparagraph*{Topological decompositions and joint crossing numbers.} 
Decomposing a surface along a graph or a curve is a standard way to simplify its topology. The classification of surfaces and classical tools to compute both homology groups and fundamental groups typically rely on such topological decompositions, which are also important in meshing and 3D-modeling (see~\cite{sheffer2007mesh}). Surfaces often come with extra structure which can be modeled by an embedded graph. Decomposing such a surface efficiently then means finding a graph that intersects the original graph transversely that does not intersect the embedded graph too much but nevertheless carries the topological complexity of the surface (see for example,~\cite{lazarus2001computing} or~\cite{erickson2005greedy}). Such decompositions also appear in algorithm design: often, to generalize results on planar graph to graphs embedded on surfaces,  its enough to find a decomposition that cuts open the surface into a disk, then solve the resulting planar instance and stitch back the solution, see, e.g.,~\cite{eh-ocsd-04,lazarus2001computing,c-tags-12}.

In many applications, it is important that the graph along which we cut is canonical in some way. For example, in order to compute a homeomorphism between two surfaces, a common approach is to cut them into disks, put these disks in correspondence, and glue back the surfaces so as to obtain a homeomorphism. However, this only works if the cut graphs have the same combinatorial structure. A seminal result on topological decompositions was pionereed by Lazarus, Pocchiola, Vegter and Verroust~\cite{lazarus2001computing} (see also~\cite{lazarushdr}) who designed an algorithm that finds, for any graph $G$ embedded in a closed orientable surface $S$ a \emph{canonical system of loops} $H$ such that no edge of $H$ intersects\footnote{Throughout the article, we decompose surface-embedded graphs by cutting them along embedded graphs which are transverse to the original graph, and count the number of intersections. This is equivalent to the primal setting studied in, e.g., Lazarus, Pocchiola, Vegter and Verroust~\cite{lazarus2001computing} via graph duality.} any edge of $G$ more than a constant number of times. Here by a canonical system of loops we mean a one-vertex and one-face embedded graph in which the cyclic ordering of the edges around the vertex is $a_1 b_1 a_1^{-1} b_1^{-1}\dots a_g b_g a_g^{-1} b_g^{-1}$.

Such a decomposition is an instance of the problem of finding simultaneous embeddings for two graphs on a surface such that the number of crossings between the two graphs is minimized. More precisely, consider a pair of graphs $G_1$ and $G_2$ embedded on a closed surface $S$ of genus $g$ and define the \emph{joint crossing number} as the minimal number of crossing points between $h(G_1)$ and $G_2$ over all the homeomorphisms $h:S\to S$. This quantity was initially introduced by Negami~\cite{negami2001crossing} who proved that any two graphs $G_1$ and $G_2$ embedded on a closed surface of genus $g$, have joint crossing number $O(g|E(G_1)||E(G_2)|)$, where $|E(G)|$ is the number of edges in $G$. Furthermore, he made the following conjecture, which is still open.

\begin{conjecture}\label{conj:negami}
There exists a universal
constant $C$ such that for any pair of graphs $G_1$ and $G_2$ embedded on a surface $S$, the joint crossing number is at most $C|E(G_1)||E(G_2)|$.\end{conjecture}

This conjecture has been investigated further~\cite{archdeacon2001two,richter2005two,hlinveny2015hardness} and variants of this problem have appeared in various works with applications as diverse as finding explicit bounds for graph minors~\cite{geelen2018explicit} or designing an algorithm for the embeddability of simplicial complexes into $\mathbb{R}^3$~\cite{matouvsek2013untangling}. From the perspective of topological decompositions, Negami’s conjecture posits that \emph{short} decompositions of any fixed shape exist, in the sense that one can always decompose an embedded graph along a chosen topological decomposition (modeled by a second, cellularly embedded graph), in such a way that each edge of the decomposition crosses each edge of the graph $O(1)$ times.\footnote{This statement is slightly stronger than Conjecture~\ref{conj:negami} since it enforces a control on the number of crossings between each pair of edges instead of the total number of crossings, but it is equally open.} Such a bound is known for only very few shapes. Beyond orientable canonical systems of loops, Colin De Verdi{\`e}re and Erickson \cite{de2010tightening} proved the existence of a short \emph{octagonal decomposition} for orientable surfaces and provided an algorithm to compute it. We do not know of any other construction of short decompositions than those two and variants thereof. 

In particular, no short decomposition at all seems to be known for non-orientable surfaces. Even if $G_1$ is a \emph{non-orientable canonical system of loops}, that is, a system of one-sided loops with the cyclic ordering $a_1 a_1 a_2 a_2 \dots a_g a_g$ around the vertex, the best known bound for this system is $O(g |E(G_2)|)$ crossings for each edge of the decomposition (see \cite{lazarushdr}), which matches the bound claimed by Negami. Due to their extra difficulty, non-orientable surfaces have been often  somehow neglected in computational topology, but there are many reasons to want to correct this: natural models of random surfaces yield non-orientable surfaces with overwhelming probability, they appear naturally as configuration spaces in diverse contexts~\cite{ghrist2008barcodes,sethna1992order}, and insights garnered from non-orientable surfaces can sometimes also be applied in a subtle way to the orientable ones; see for example~\cite{schaefer2013block}. Furthermore the orientable genus of a graph can be arbitrarily larger than its non-orientable genus, while the reverse does not happen (see Lemma~\ref{odd non-orientable genus}).

\subparagraph*{Our results.} In this article, we initiate a thorough study of short topological decompositions on non-orientable surfaces. As outlined above, one of the only results known on topological decompositions of non-orientable surfaces is a theorem of Negami~\cite{negami2001crossing}. We first show that the proof of this result has a minor flaw and exhibit a specific counter-example to the proof technique. Then we provide an alternative proof based on different techniques.

\begin{restatable}{theorem}{Neg}\label{T:negamimain}
Let $S$ be a non-orientable surface of genus $g\geq 1$ and $G_1$ and $G_2$ be two graphs embedded on $S$. Then there exists a homeomorphism $h$ such that any edge of $h(G_1)$ crosses each edge of $G_2$ at most $O(g)$ times. In particular, the total number of crossings between $h(G_1)$ and $G_2$ is $O(g|E(G_1)||E(G_2)|)$.
\end{restatable}

In order to prove this theorem we take advantage of a technique in~\cite{matouvsek2013untangling} to compute a short \emph{orienting curve}, i.e., a curve such that cutting along it yields an orientable surface.

Our main result is the following theorem providing, to the best of our knowledge, the first known case of a short topological decomposition into a disk for non-orientable surfaces.

\begin{restatable}{theorem}{canonical}\label{canonical}
There exists a polynomial time algorithm that given a graph cellularly embedded on the non-orientable surface $N$ computes a non-orientable canonical system of loops such that each loop in the system intersects any edge of the graph in at most $30$ points.
\end{restatable}

\subparagraph*{Main ideas and proof techniques.}

As in many similar works the first step in most of our results is to contract a spanning tree of the underlying graph, reducing the problems to the setting of one-vertex graphs embedded on a non-orientable surface. The combinatorics of a one-vertex embedded graph are completely described by a rotation system, or embedding scheme, i.e., by the circular order of the edges around the vertex, and a signature for each loop indicating whether it is one-sided or two-sided. Such an embedding scheme will be the basic object with which we work.

Then, a simple but important object that we rely on extensively is an \emph{orienting curve}, i.e., a closed curve on a non-orientable surface such that cutting along it produces an orientable surface. It was shown by Matou\v{s}ek, Sedgwick, Tancer and Wagner~\cite{matouvsek2013untangling} that given a graph embedded on a non-orientable surface, one can compute such an orienting loop that crosses each edge of the graph at most a constant number of times. This lemma allows us, at the cost of slightly increasing the constants in our results, to assume that our embedding schemes always have an orienting curve. With this tool at hand we provide a corrected proof of Theorem~\ref{T:negamimain}. Furthermore, the existence of this orienting curve will significantly simplify our work to prove Theorem~\ref{canonical}.

For the proof of Theorem~\ref{canonical}, we first point out that the techniques used to prove the orientable version in~\cite{lazarus2001computing} do not readily apply, as they rely on a fine control of the cut-and-pasting operations used in the proof of the classification of surfaces, and in the non-orientable case there is an additional step in these operations which
incur an overhead of $O(g)$ in the number of crossings of the resulting curves (see~\cite[Theorem~4.3.9]{lazarushdr}). Instead, our proof of Theorem~\ref{canonical} builds on important recent work of Schaefer and \v{S}tefankovi\v{c}~\cite{JGAA-580}. The foundational idea behind this work, which takes its roots in an article of Mohar~\cite{mohar2009genus} on the degenerate crossing number, is to represent a graph embedded on a non-orientable surface as a planar drawing, with a finite set of \emph{cross-caps}, which are points where multiple edges are allowed to cross in a way that reverses the permutation, as pictured in the second picture of Figure~\ref{F:crosscaps}. Using an intricate argument inducting on the loops of an embedding scheme, Schaefer and \v{S}tefankovi\v{c} showed that any graph embedded on a non-orientable surface can be represented in such a cross-cap drawing so that each edge uses each cross-cap at most twice. Our main technical contribution is to upgrade their construction so that the cross-caps can be connected to each other so as to yield a non-orientable canonical system of loops (Lemma \ref{right combinatorics}), so that each loop intersects each edge of the one-vertex graph in at most $30$ points (see Figure~\ref{F:crosscaps}).

\begin{figure}[t]
    \centering
    \includegraphics[width=\textwidth]{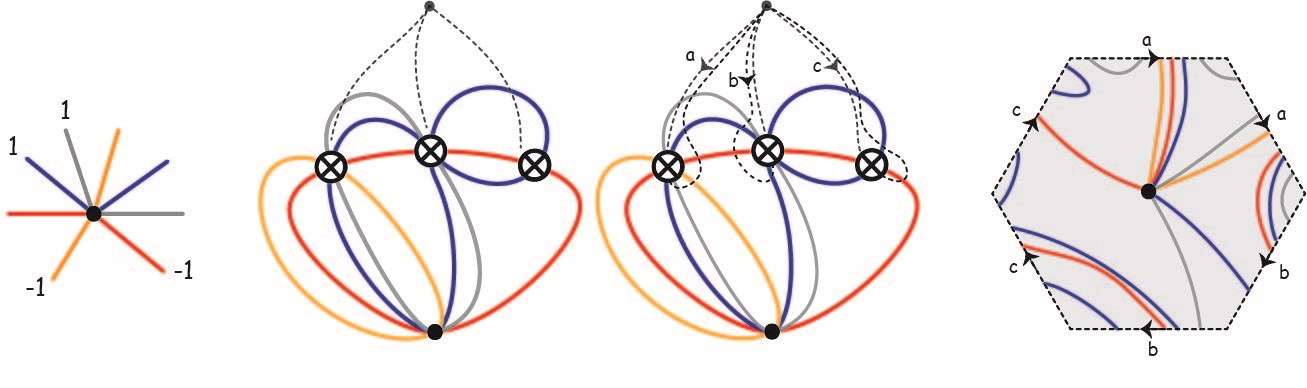}
    \caption{From left to right: 1)The combinatorial information of a one-vertex graph. 2) A cross-cap drawing of this graph, with cross-caps connected to a base-point. 3) A joint drawing of the graph and a canonical system of loops. 4) A different representation: decomposing the graph with a cannonical system of loops.}
    \label{F:crosscaps}
\end{figure}

The complexity of the drawings provided by the proof of Schaefer and \v{S}tefankovi\v{c} increases too fast to directly obtain a good bound by just connecting the cross-caps. Therefore, we modify their algorithm. First, by the aforementioned techniques, we can assume that we always have an orienting loop, which simplifies some of the steps and provides additional structure to the inductive argument. But more importantly, we show that one can impose a certain order in which we choose the one-sided loops, as well as the separating loops, in the inductive argument of Schaefer and \v{S}tefankovi\v{c} so as to obtain a finer control on the resulting drawing. 

The order in which we choose loops comes from a seemingly unrelated problem in computational biology, and more precisely genome rearrangements. Given a permutation with signatures (a bit assigned to each letter), a signed reversal consists in choosing a subword in $w$, and reversing it as well as the signatures of all its letters. The \emph{signed reversal distance} between two signed permutations is the minimum number of signed reversals needed to go from one permutation to the other one. This distance, and in particular algorithms to compute it has been intensively studied in the computational biology literature due to its relevance for phylogenetic reconstruction (see for example~\cite{hayes2007computing}). A cornerstone of the theory is the breakthrough of Hannenhalli and Pevzner~\cite{hannenhalli1999transforming}  who provided an algorithm to compute the signed reversal distance between two signed permutations in polynomial time (see also the reformulation by Bergeron~\cite{bergeron2001very}). Now, as we illustrate in Figure~\ref{F:pancakes}, there is a very strong similarity between computing the signed reversal distance between two permutations and embedding a one-vertex graph built from these two permutations with a minimum number of cross-caps (see~\cite{bura2016lower,huang2017topological}). Surprisingly the algorithms of Hannenhalli and Pevzner on one side and of Schaefer and \v{S}tefankovi\v{c} on the other also have similarities. The proof of Theorem~\ref{canonical} leverages the literature on the signed reversal distance problem, and in particular the structure we impose on the cross-caps drawings of non-orientable graphs is inspired on ideas from the aforementioned genome rearrangements algorithms. We hope that further interpollination between computational genomics and computational topology will lead to new surprises.\\

\begin{figure}[t]
    \centering
    \includegraphics[width=.80\textwidth]{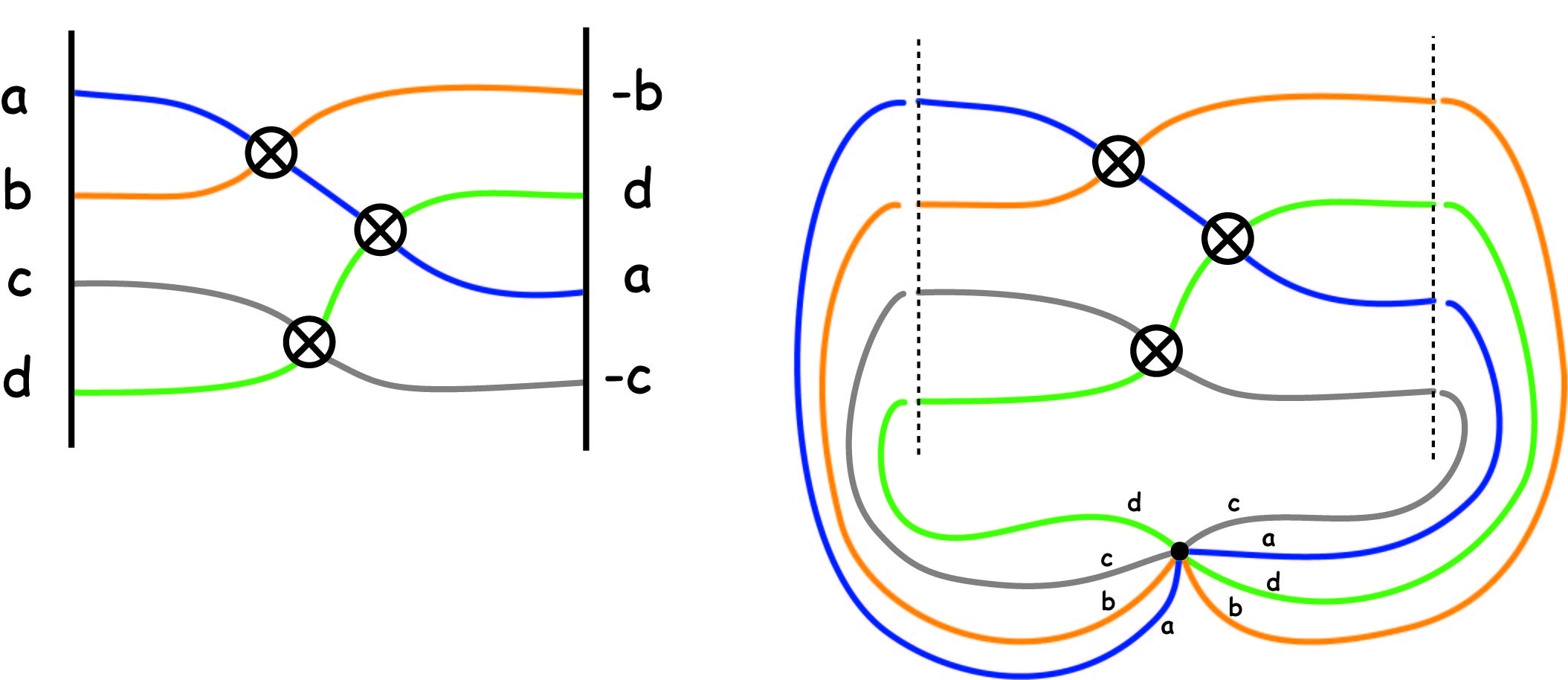}
    \caption{Left: a pictorial representation of three signed reversals bringing the signed permutation on the left to the signed permutation of the right. Right: Attaching the two permutations to a common basepoint yields a one-vertex graph with an embedding scheme, and the signed reversals provide a cross-cap drawing of that scheme where each loop enters each cross-cap at most once.}
    \label{F:pancakes}
\end{figure}

\subparagraph*{Outline.}
After introducing the preliminary definitions and results in Section \ref{pre}, we will prove Theorems \ref{T:negamimain} and \ref{canonical} 
 in Sections \ref{negami-correction} and \ref{non-canonical} respectively. 

\subparagraph*{Acknowledgements.} We are grateful to Marcus Schaefer and Daniel \v{S}tefankovi\v{c} for providing us the full version of~\cite{JGAA-580}, to Francis Lazarus for insightful discussions, and to the anonymous reviewers for very helpful comments.

\section{Preliminaries}\label{pre}

While this paper strives to be mostly self-contained, we refer the reader to standard references such as Hatcher~\cite{hatcher2002algebraic} and Stillwell~\cite{stillwell1993classical} for more topological background, the book of Mohar and Thomassen for an extensive overview of graphs on surface~\cite{mohar2001graphs} and the survey of Colin de Verdière~\cite{c-ctgs-18} on topological algorithms for embedded graphs.

\subparagraph*{Surfaces.} A \emph{surface} $S$ is a topological Hausdorff
space where each point has a neighborhood homeomorphic to either the plane or
the closed half-plane. The points without a neighborhood homeomorphic to the plane
comprise the boundary of $M$. Compact surfaces without boundaries are called \emph{closed surfaces}. A surface is called \emph{orientable} if it does not contains a subspace homeomorphic to a M\"{o}bius band; otherwise, it is called \emph{non-orientable}. 
Throughout this work, we denote orientable surfaces and non-orientable surfaces by $M$ and $N$ respectively, and by $S$ whenever orientability did not make a difference.

\emph{The classification theorem} for closed surfaces states that any orientable surface $M$ of genus $g\geq0$ is homeomorphic to a sphere with $g$ handles\footnote{A \emph{handle} is obtained by removing a small disk and gluing a punctured torus along its boundary to the boundary circle of the resulting hole (see the left picture in Figure \ref{classification}).}, $g$ is called the \emph{orientable genus} of $M$; and any non-orientable surface $N$ of genus $g\geq 1$ is homeomorphic to a sphere with $g$ cross-caps\footnote{A \emph{cross-cap} is obtained by removing a small disk from the sphere and gluing in a M\"{o}bius band along its boundary to the boundary circle of the resulting hole (see the right picture in Figure \ref{classification}).}, in this case, $g$ is the \emph{non-orientable genus} of $N$ (see Figure \ref{classification}). Throughout this work, by the \emph{genus} of a surface, we mean the orientable genus for orientable surfaces and the non-orientable genus for non-orientable surfaces. If $S$ is a surface with $k$ boundary components, the classification theorem still applies, except that we need to replace the initial sphere with a sphere with $k$ holes. Thus a surface is uniquely determined by its genus, by its number of boundary components, and by its orientability. 
\begin{figure}[t]
\centering
\includegraphics[width=0.60\textwidth]{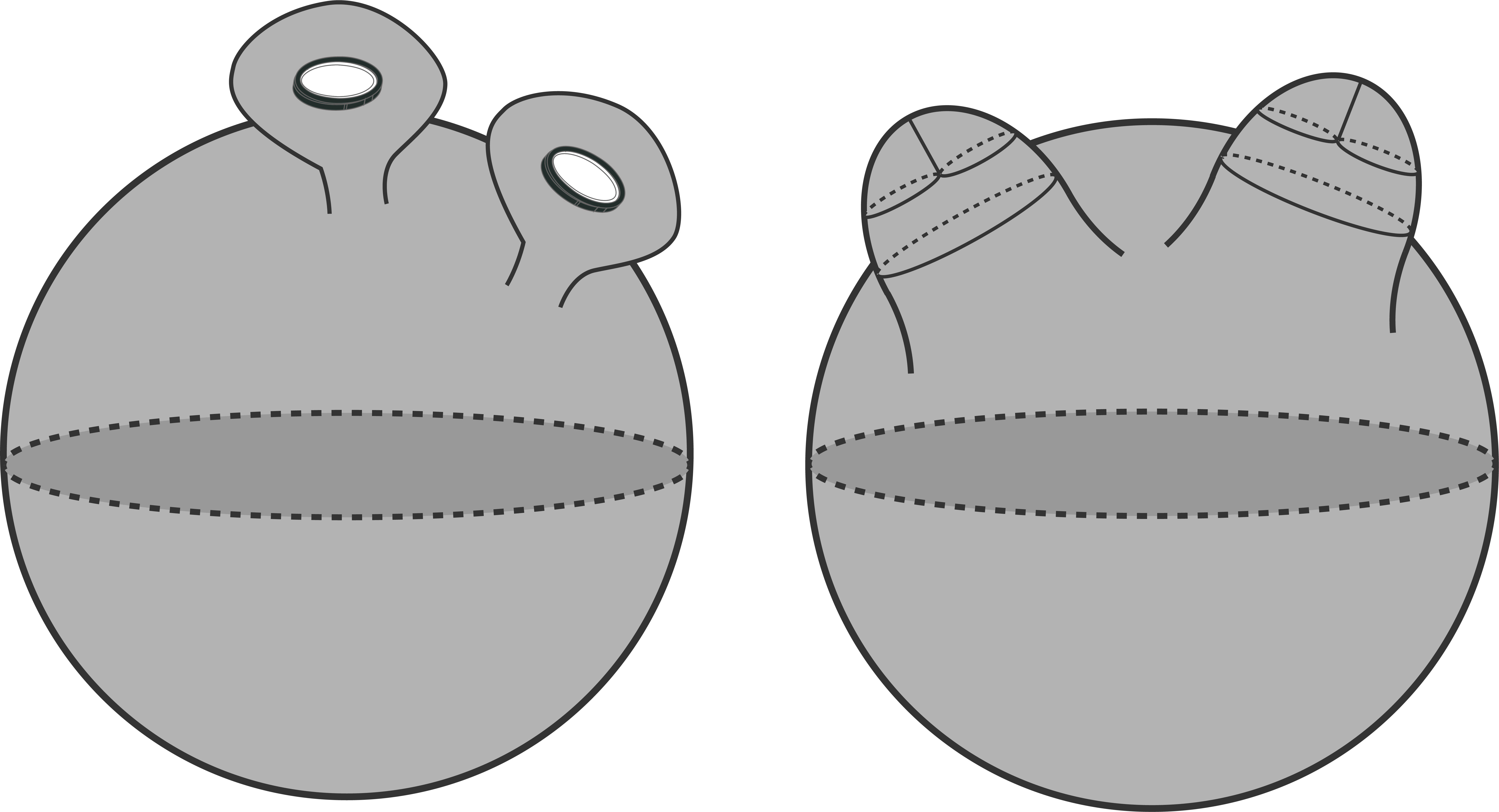} 
\caption{Depiction of compact surfaces without boundaries obtained by attaching handles (orientable surfaces, left picture) or cross-caps (non-orientable surfaces, right picture).}
\label{classification}
\end{figure}

\subparagraph*{Curves and Embedded Graphs}

A \emph{closed curve} or a \emph{cycle} on a surface $S$ is a continuous map $\theta: S^1\to S$. By a \emph{path} from $x$ to $y$ on a surface, we mean a continuous map $\theta: [0,1]\to S$ where $\theta(0)=x$ and $\theta(1)=y$. A path with two ends on the boundary of a surface is called an \emph{arc}. These are called simple if the maps are injective.

We call a closed curve on a surface \emph{two-sided} if a small closed neighborhood of it is homeomorphic
to the annulus. Otherwise, it is called \emph{one-sided} and it has a closed neighborhood homeomorphic to the M\"{o}bius band. 
Given a closed curve $\nu$ on a surface $S$, cutting $S$ along $\nu$ gives a (possibly disconnected) surface with one or two boundary components depending on whether $\nu$ is one-sided or two-sided. A curve $\delta$ on a surface $S$ is called \emph{non-separating} if the surface we obtain by cutting along $\delta$ is connected; otherwise $\delta$ is separating. An \emph{orienting curve}\footnote{This terminology is slightly non-standard, we follow Schaefer and \v{S}tefankovi\v{c}~\cite{JGAA-580} who attribute it to B. Mohar} on a non-orientable surface $N$ is a curve $\gamma$ such that by cutting along $\gamma$, we get a connected orientable surface. We recall the following lemma from \cite{matouvsek2013untangling}.
\begin{lemma}\label{matouvsek2}
(\cite{matouvsek2013untangling}, Lemma 5.3). Let $N$ be a non-orientable surface of genus $g$ with $h$ boundary components and let $\gamma$ be an orienting closed curve. Let $g_\gamma$ be the (orientable) genus and $h_\gamma$ be the number of boundary components in $N$ after cutting along $\gamma$.
\begin{itemize}
    \item If $g$ is odd, then $\gamma$ is one-sided, $g_\gamma = \frac{g-1}{2}$, and $h_\gamma = h+1$
    \item If $g$ is even, then $\gamma$ is two-sided, $g_\gamma = \frac{g-2}{2}$, and $h_\gamma= h + 2$.
\end{itemize}
\end{lemma}

Two paths $p$ and $q$ with the same end points $a$ and $b$ on a surface $S$, are \emph{homotopic} if there is a continuous map $H:[0,1]\times [0,1]\to S$ such that $H(0,.)=p$ , $H(1,.)=q$,
$H(.,0) = a$, and $H(.,1) = v$. Two cycles $\gamma$ and $\gamma^{'}$ are freely homotopic if there is a continuous map $H: [0, 1]\times S^1 \to S$ such that
$H(0, t) = \gamma(t)$ and $H(1, t) = \gamma^{'}(t)$ for all $t$. A cycle is \emph{contractible} if it is homotopic to a
constant cycle; an arc is \emph{contractible} if it is homotopic to a path on the boundary. 

\subparagraph*{Graph Embeddings.}
In a drawing of a graph on a surface, every vertex is mapped to a point on the surface and edges are realized as curves connecting their endpoints. Informally, a drawing of the graph $G$ on a surface $S$ with no crossing on the edges is called an embedding of $G$ on $S$. More precisely, an embedding of $G$ is a continuous, one-to-one map from $G$ into $S$. In this work, we generally identify $G$ with its embedding on $S$. A graph embedding is called \emph{cellular} if its faces are homeomorphic to open disks. For a graph $G$ cellularly embedded on a surface $S$, \emph{Euler's formula} states that $v-e+f=\chi(S)$, where $v$, $e$ and $f$ represent the number of vertices, edges and faces of $G$ and $\chi(S)$ is the \emph{Euler characteristic}, a topological invariant that depends only on the surface and not on the cellular embedding. The \emph{Euler genus} of a closed surface is defined by $eg(S)=2-\chi(S)$. It is easy to see that the Euler genus is twice the orientable genus for orientable surfaces, and equal to the non-orientable genus for non-orientable ones. We take these identities as definition for surfaces with boundary, then the Euler characteristic of a surface with boundary is $2-eg(S)-k$ where $k$ is the number of the boundary components of the surface. 
Cutting a surface along an arc $\alpha$ increases the Euler characteristic of the surface by 1 and cutting along a closed curve does not change the Euler characteristic. These relations allow us to relate the genus and the number of boundary components of a surface after a cutting. 

\begin{lemma}\label{adding separating}
Let $t$ be a separating (closed) curve on a closed surface $S$ and let $S_1$ and $S_2$ be the surfaces we obtain by cutting along $t$. We have $eg(S)=eg(S_1)+eg(S_2)$.
\end{lemma}

\begin{proof}
By cutting along $t$, we get two boundary components, one on each of $S_1$ and $S_2$ and since $t$ is a closed curve, we have $2-eg(S)=(2-eg(S_1)-1)+(2-eg(S_2)-1)$ which proves the claim. 
\end{proof}
\subparagraph*{Discrete Metric on Surfaces.} We briefly define the notions of combinatorial and cross-metric surfaces, which define a
discrete metric on a surface and are essentially the same, up to duality (see \cite{de2010tightening} for more details). For a graph $G$ cellularly embedded on a surface without boundary, the \emph{dual graph} $G^*$ is defined as follows: to each face of $G$ we associate a vertex and for each vertex $e$ in $G$ that separates the faces $f_1$ and $f_2$, we connect the vertices associated to $f_1$ and $f_2$.  Note that in this construction, two vertices can be connected with multiple edges if they have more than one common edge. A dual graph embedding is also cellular and $G^{**}=G$.

A \emph{combinatorial surface} is a surface $S$ together with a graph $G$ which is cellularly embedded on $S$. In the case that $S$ has boundary, $G$ is embedded so that the boundary is the union of some edges of $G$.
In this model, the only allowed curves are walks in $G$,
and the length of a curve $C$ is the number of edges of $G$ traversed by $C$. A \emph{cross-metric surface} is a surface $S$ together with a graph $G$ which is cellularly embedded on $S$. In the case that $S$ has boundary, $G$ is embedded so that the boundary is the union of some edges of $G$. The curves allowed in this definition, are those that cross $G$ only transversely (i.e., that intersect $G$ only at edges and in a non-tangent way), and the length of a curve $C$ is the number of edges of $G$ that $C$ crosses.

In this work, we mostly work in the cross-metric model and we refer to the embedded graph $G$ of the cross metric surface $S$, as the primal graph on $S$. The \emph{multiplicity} of a curve (or a system of curves) at some edge $e$ of $G$ is the number of times $e$ is crossed by the curve (curves). The
multiplicity of a curve (or a system of curves) is the maximal multiplicity of the curve (curves) at any edge $e$ of $G$.

\subparagraph*{Embedding schemes.} For $v$ a vertex of an embedded graph $G$, by a \emph{rotation} $\rho_v$ at $v$, we mean the cyclic permutation of the ends of edges incident to $v$. A \emph{rotation system}, $\rho$, of a
graph assigns a rotation to each vertex. We assign to each edge a signature which is a number from $\{1,-1\}$.
A rotation system $\rho$ and a signature $\lambda$ for the edges determine a cellular embedding for the graph up to homeomorphism i.e. we can compute the faces of the embedding purely combinatorially (see \cite{mohar2001graphs} for further details). The pair $(\rho,\lambda)$ is called an \emph{embedding scheme} for the graph $G$, we simply use \emph{scheme} instead of embedding scheme throughout this work. Since a first step in all of our arguments is to contract a spanning tree, almost all the embedding schemes considered in this article will have a single vertex. A cycle in a scheme is one-sided if the signature of its edges multiply to $-1$ and it is two-sided otherwise.

 A scheme is \emph{orientable} if all its cycles are two-sided, and \emph{non-orientable} otherwise.
A loop $e$ in the scheme divides the half-edges around the vertex into two parts; each part is called a \emph{wedge} of $e$. When a loop $g$ has exactly one end in each wedge of $e$, we say that the ends of $g$ \emph{alternate} with those of $e$; otherwise both ends of $g$ is in one wedge of $e$ and we say that the ends of $e$ \emph{enclose} the ends of $g$. Finally, we slightly depart from the convention in topological graph theory: we do \textbf{not} assume that a scheme describes a cellularly embedded graph, since we will sometimes consider orientable schemes over non-orientable surfaces. However, we will always consider embedded graphs (and their schemes) on surfaces of minimal non-orientable genus, so while these graphs are sometimes not cellularly embedded, they have at most one non-cellular face with non-orientable genus one (see Lemma~\ref{odd non-orientable genus}).

 Following Schaefer and \v{S}tefankovi\v{c}~\cite{JGAA-580}, we will use the following model with localized cross-caps to represent non-orientable embedded graphs. A \emph{planarizing system of disjoint one-sided curves} on a non-orientable surface, abbreviated \emph{PD1S}, is a system of $g$ disjoint one-sided curves such that by cutting along them, we obtain a sphere with $g$ holes (this was first introduced by Mohar~\cite{mohar2009genus}). Therefore, from any graph embedded on a non-orientable surface, we obtain a planar representation by cutting along such a system. The non-orientable surface is recovered by gluing a M\"{o}bius band on each boundary component, which we depict using $\tiny{\bigotimes}$ and call a \emph{cross-cap}. It is easily checked that a family of edges entering a cross-cap emerge on the other side with a reversed order, and that the sidedness of a loop is determined by the number of cross-caps that it crosses.
 
The planar drawing that we obtain by this cross-cap localization is called a \emph{cross-cap drawing}, see Figure~\ref{equivalent} for examples. In this model, we say that a drawing realizes an embedding scheme $(G,\rho,\lambda)$ if the rotation at each vertex is as prescribed by $\rho$, and if whenever a closed curve in the drawing passes through an odd (resp. even) number of cross-caps, the multiplication of the signatures of the edges it follows is  $-1$ (resp. $1$).  
Note that while a cross-cap drawing uniquely describes an embedded graph, the converse is not true, see Figure~\ref{equivalent}. 
 Throughout this article, by a cross-cap drawing for a graph $G$ with an embedding scheme $(\rho,\lambda)$, we mean the planar graph with cross-caps treated as extra vertices and edges being the sub-edges in $G$. 

 \begin{figure}[t]
    \centering
    \includegraphics[width=.75\textwidth]{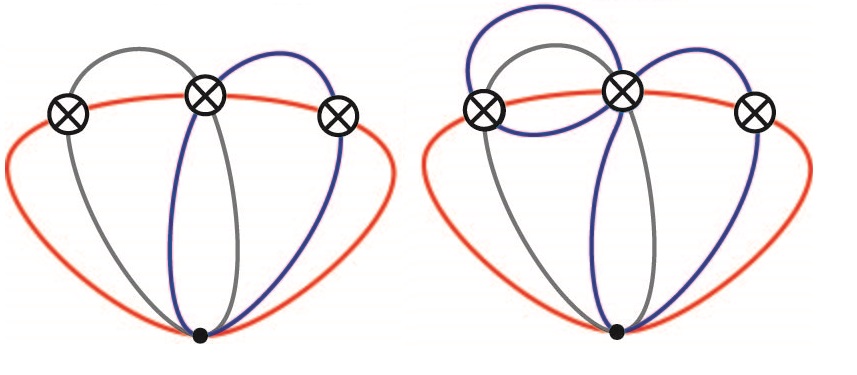}
     \caption{Different localization for the same embedding scheme gives different cross-cap drawings}
    \label{equivalent}
\end{figure}

  The following lemma helps us recognize an orienting loop in a one vertex scheme.

 \begin{lemma}\label{orienting loops}

 A loop $o$ in a cellularly embedded one-vertex graph $G$ with a non-orientable embedding scheme is orienting if and only if its ends enclose the ends of any two-sided loop and alternate with the ends of any one-sided loop in the scheme.
\end{lemma}

\begin{proof}
For the direct implication, let $o$ be an orienting loop and consider a cross-cap drawing of $G$ with a minimal number of cross-caps. By the classification of surfaces and Lemma~\ref{matouvsek2}, there is, up to homeomorphism, a single orienting loop, and therefore we can assume that in this cross-cap drawing the loop $o$ goes exactly once through each cross-cap. Therefore $o$ separates the plane into two regions, and since two-sided loops cross an even number of cross-caps, they start and end in the same wedge formed by $o$. Likewise, since one-sided loops cross an odd number of cross-caps, they start and end in a different wedge, and thus alternate with $o$.

For the reverse implication, we investigate the surface and the graph obtained after cutting along $o$. If $o$ is one-sided, the resulting surface has one boundary, and the vertex is split into two vertices, one on each boundary. If $o$ is two-sided the resulting surface has two boundaries, and the vertex is split into two vertices on both boundaries. In both cases, loops which were alternating with $o$ are now arcs connecting two distinct vertices, while loops which were not alternating connect the same vertex. After the cutting, edges keep their signature, and thus we see that edges connecting the two vertices have signature $-1$ and therefore any cycle in the resulting graph has signature $1$. Thus the resulting embedding scheme is orientable, and since the graph embedding was cellular before cutting, it still is after cutting, and the resulting surface is orientable. Therefore $o$ is orienting.
\end{proof}
 
 The following lemma helps us identify separating and orienting loops in cross-cap drawings.
 
  \begin{lemma}\label{homology}
 In any cross-cap drawing of a scheme, a separating loop passes through each cross-cap an even number of times; and an orienting loop passes through each cross-cap an odd number of times. 
 \end{lemma}

 \begin{proof}
 Actually this statement does not depend on the whole scheme, only on the loops. It is enough to show that: \begin{itemize}
     \item A cross-cap drawing of a separating closed curve passes through each cross-cap an even number of times.
     \item A cross-cap drawing of an orienting closed curve passes through each cross-cap an odd number of times. 
 \end{itemize}
 
 Observe that if $\gamma$ is one-sided it cannot be separating, moreover $\gamma$ can separate the surface into at most two connected components, one for each side of $\gamma$ and if $\gamma$ is separating every time we cross $\gamma$ we should change connected component.
 Now consider a cross-cap drawing of $\gamma$, and assume to reach a contradiction that some cross cap is crossed an odd number of times by $\gamma$. Choose any wedge around that cross cap and notice that if we go around the cross-cap with a curve $\gamma'$, to reach the opposite wedge defined by $\gamma$, then $\gamma'$ and $\gamma$ intersect an odd number of times, so they should be in different connected components of the complement of $\gamma$. This is a contradiction since opposite wedges are clearly on the same connected component.\\

 Let $\gamma$ be an orienting curve, and denote by $\alpha_1, \ldots, \alpha_g$ the PD1S underlying the cross-cap drawing. By the classification of surfaces, $\gamma$ is unique up to homeomorphism, i.e., there exists (another) cross-cap drawing where $\gamma$ goes exactly once through every cross-cap. The image of a curve $\alpha_i$ under this homeomorphism is a simple closed curve on the same surface, which in this new cross-cap drawing intersects $\gamma$ and the new cross-caps. Furthermore, it intersects the new cross-caps an odd number of times since it is one-sided. Now it is immediate that in this new representation, $\gamma$ partitions the plane into two regions, and any curve crossing the cross-caps an odd number of times must also cross $\gamma$ an odd number of times. This proves the needed property.
 \end{proof}

\begin{lemma}\label{odd non-orientable genus}
Let $G$ be an orientable scheme corresponding to a cellular embedding on a surface $M$ with a minimum number of $g>0$ handles. The minimum number of cross-caps needed in a cross-cap drawing realizing $G$ is $2g+1$ and this can always be achieved. 
\end{lemma}

\begin{proof}
The scheme has Euler genus $2g$, hence at least $2g$ cross-caps are required. Now Lemma~6 in \cite{JGAA-580} states that any cross-cap drawing of an orientable embedding scheme of genus $g\neq 0$ requires an odd number of cross-caps, so at least $2g+1$ are required.
To see that this always suffices we begin with the embedding on the orientable surface and add a cross-cap in one of the faces of $G$ on $M$; this gives us an embedding of $G$ on a non-orientable surface with one cross-cap and $g$ handles that is homeomorphic to the non-orientable surface $N$ with $2g+1$ cross-caps.\end{proof}

\subparagraph*{Canonical System of Loops.} For an orientable surface $M$ of genus $g$, we define the \emph{orientable canonical system of loops} to be a family of two-sided loops with the cyclic ordering $a_{1}b_{1}a_{2}b_{2}\dots a_{g}b_{g}a_{g}b_{g}$ around the base point, such that cutting $M$ along this family yields a topological disk. For a non-orientable surface of genus $g$, the \emph{non-orientable canonical system of loops} is a family of one-sided loops with the cyclic ordering $a_{1}a_{1}a_{2}a_{2}\dots a_{g}a_{g}$ around the base point such that cutting $M$ along this family yields a topological disk. The following lemma underpins our strategy to prove Theorem~\ref{canonical}: in order to find a non-orientable canonical system of loops, first find a cross-cap drawing and connect the cross-caps to a root using short paths (see Figure~\ref{F:crosscaps}. 

\begin{lemma}\label{right combinatorics}
Let $H$ be a cross-cap drawing for a graph of non-orientable genus $g$ and let $b$ be a point in one face of the drawing. Let $\{p_i\}$ be a family of paths in the dual graph to this drawing from each cross-cap to $b$. Introduce a loop $c_i$ by starting from $b$, passing along the path $p_i$, entering the corresponding cross-cap, going around the cross-cap and passing along $p_i$ to return to $b$. The system of loops $\{c_i\}$ is a non-orientable canonical system of loops. 
\end{lemma}

\begin{proof}
It is easy to check that each $c_i$ has consecutive ends around $b$. Each curve $c_i$ is homotopic to a concatenation of a curve in a planarizing system of disjoint one-sided curves and two copies of a path $p_i$, therefore it is one-sided. Cutting along these system of curves corresponds to cutting along a planarizing disjoint one-sided loops which gives us a sphere with $g$ boundary components and then cutting along the paths $\{p_i\}$ which connect these boundary components and cut them to a single boundary. Therefore cutting along $\{c_i\}$, cuts the surface into a disk, and we obtain a non-orientable canonical system of loops. This is illustrated in Figure~\ref{F:crosscaps}.
\end{proof}

\subparagraph*{Short Orienting Curves.} Throughout this paper, to deal with non-orientable surfaces, we turn them to an orientable surface by cutting along an orienting curve. For this approach to work in our problems, we need to ensure that we can find an orienting curve that crosses our primal graph not too many times. The following lemma is a restatement of the Proposition 5.5 in~\cite{matouvsek2013untangling}. We provide a sketch of proof, explaining how to extend it to arbitrary embedded graphs and how to modify it to get orienting arcs in the presence of boundaries.

 \begin{lemma}\label{matouvsek}
Let $N$ be a non-orientable surface without boundary and with genus $g$ and $G$ be a graph embedded on $N$. Then there exists an orienting curve of multiplicity at most $2$.
\end{lemma}

\begin{proof}[Sketch of proof]
We begin by adding edges to the embedded graph $G$ in order to get cellular faces. We assign a local orientation to each face, that is a cyclic order to the vertices of each face along its boundary. Two adjacent faces are said to have \emph{incoherent} orientations if they induce the same orientation on the edge they have in common. Cutting the surface along all these edges gives us an orientable surface. For any vertex $v$ there is an even number of edges with incoherent orientations adjacent to $v$. We pair these edges around each vertex, shift them slightly so that they cross the original graph only transversely and we add segments to join two paired edges. We can modify our local orientation slightly to show that cutting along our new system of edges gives us an orientable surface; see Figure \ref{m1}.
\begin{figure}[t]
    \centering
    \includegraphics[width=0.58\textwidth]{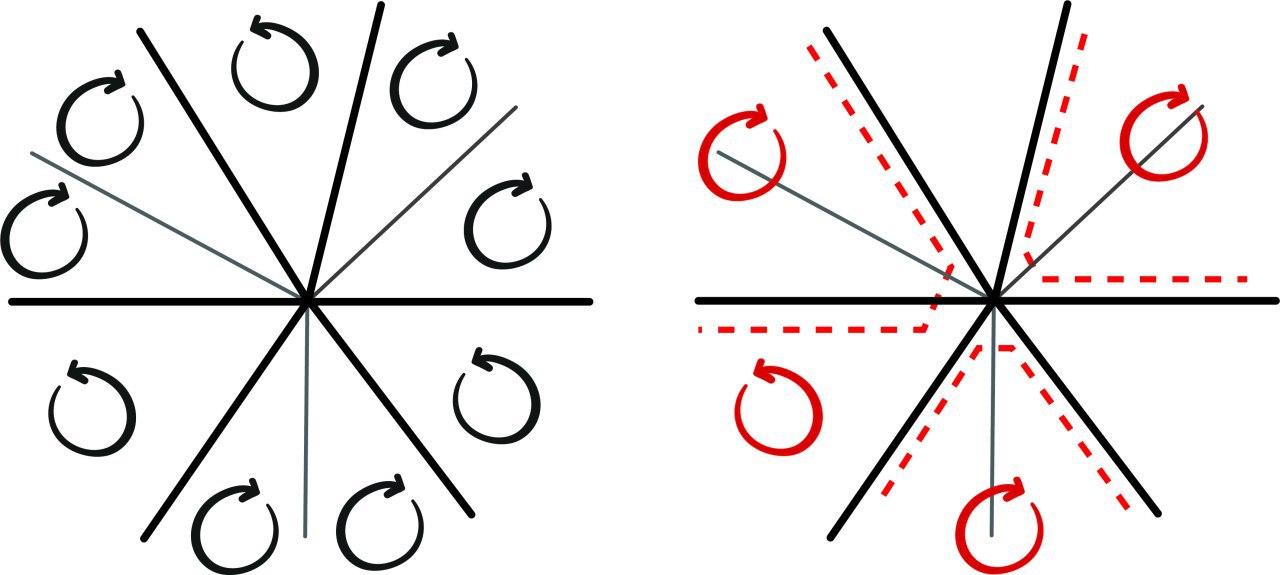}
    \caption{The thick lines in the left picture, depict the edges adjacent to a vertex that inherit the same orientation from different faces. The dashed red lines in the right picture are shortened and shifted copies of these edges joined according to the pairing. The arrows indicate the refinement of the local orientation.
    }
    \label{m1}
\end{figure}

The graph we obtain by these new edges is a disjoint union of closed curves. Each of these curves intersects each edge of the original graph at most once near the vertices, therefore the whole system of curves has crosses each edge of $G$ at most twice. If we have one curve, then this curve is the desired orienting curve. In the case we have more than one curve, we can find short paths (paths that do not intersect edges that are crossed already by the curves) that joins pairs of these curves at each step and slightly changing the orientations around these paths, we can obtain a concatenation of these curves and paths to only one curve. For the complete proof see \cite{matouvsek2013untangling}.
\end{proof}

\begin{remark}\label{remarkorienting}
   If $N$ is a non-orientable surface with a boundary component, by a little modification of the building process in the proof of the lemma, we can get an orienting arc instead of an orienting cycle. Note that if the surface has boundary, the local orientation around each boundary is similar to those of a vertex. Therefore, we can choose one boundary component and shift and join all the edges that inherit the same orientation, except one pair. This way we get an arc and a system of closed curves and we can proceed as explained in the proof to join these components and get one orienting arc.
\end{remark}

\section{Correcting the proof of Negami}\label{negami-correction}

Negami proved in~\cite{negami2001crossing} that if we have two graphs embeddable on a closed surface, we can reembed them simultaneously such that their edges cross few times.

\begin{theorem}\label{NEGAMI}
(\cite{negami2001crossing}, Theorem 1). Let $G_1$ and $G_2$ be two connected graphs embeddable on a closed surface of genus $g$, orientable or non-orientable. We can embed them simultaneously such that they intersect transversely in their edges at most $4g\beta(G_1)\beta(G_2)$ times, where $\beta(G)=|E(G)|-|V(G)|+1$ is the Betti number for a connected graph $G$.
\end{theorem}

Although the statement in Theorem \ref{NEGAMI} is correct (up to a constant factor), his proof is incorrect for non-orientable surfaces. In this section, first we introduce a bad example illustrating the erroneous approach in Negami's proof in the non-orientable case and then we correct the proof by proving the following theorem.\footnote{This theorem is stated with respect to the numbers of edges $|E(G_i)|$ for simplicity, but the proof also provides a bound in terms of the Betti numbers $\beta(G_i)$, as in Theorem~\ref{NEGAMI}.}

\Neg*

 The proof that Negami proposed for Theorem \ref{NEGAMI} is reduced to showing the following lemmas (with slightly different constants). While the proof of Lemma~\ref{negami theorem} is fine, we argue that the proof technique behind Lemma~\ref{correction of negami} is flawed. By an \emph{essential} proper arc, we mean an arc with endpoints on a boundary component that does not cut off a disk from the surface.

\begin{lemma}\label{negami theorem}
For two orientable surfaces $M_i$ of genus $g\geq1$, with one boundary component and $\beta_i$ disjoint essential proper arcs ($i=1,2$) where $\beta_i\leq \beta(G_i)$, there are homeomorphisms $\phi_i: M_i\to M$, where $M$ is an orientable surface of genus $g$ and one boundary component, so that the image of the arcs in $M_1$ and $M_2$ on $M$ intersect at most $4(g-1)\beta_1\beta_2$ times.
\end{lemma}

\begin{lemma}\label{correction of negami}
For two non-orientable surfaces $N_i$ of genus $g\geq1$  with one boundary component and $\beta_i$ disjoint essential proper arcs ($i=1,2$) there are homeomorphisms $\phi_i: N_i\to N$, where $N$ is a non-orientable surface of genus $g$ and one boundary component, so that the image of the arcs in $N_1$ and $N_2$ on $N$ intersect at most $18(g-1)\beta_1\beta_2$ times when $g$ is odd and $72(g-2)\beta_1\beta_2$ when it is even.
\end{lemma}

In the proof of Lemma~\ref{correction of negami}, Negami uses induction on the genus of the non-orientable surface. Assuming that the claim is true for genus $g-1$, to prove it for genus $g$, he claims that there is an essential proper arc $\alpha$ that runs along the center line of a M\"{o}bius band (a one-sided arc). This arc can be either included in the system of arcs or be disjoint from it. The idea is then to cut along $\alpha$ to get a non-orientable surface of genus $g-1$ to use the induction hypothesis. The problem lies in the fact that cutting along such an arc, we might not end up with a non-orientable surface. This is illustrated in the following lemma.

\begin{lemma}
Consider the non-orientable surface of genus 3 with one boundary component and embedded essential arcs shown in Figure \ref{counterexample}. Any one-sided arc disjoint from the embedded arcs cut the surface into an orientable surface.
\end{lemma}

\begin{proof}
Figure~\ref{counterexample} depicts a non-orientable surface of genus $3$ (obtained by identifying the boundary edges according to their letters and orientations) and one boundary component, and a family of essential arcs on it (consisting of the boundary edges and the blue arcs). The arc $c$ is a one-sided arc and $a$ and $b$ are two non-homotopic two-sided arcs. The blue arcs are all two-sided arcs embedded on the surface. A one-sided arc disjoint from the system of arcs in this polygonal schema, must have one end on the segment of the boundary component between two copies of $c$ and the other end on one of the two other segments adjacent to $c$. By cutting along such an arc, we get an orientable surface of Euler genus 2 with one boundary (such an arc is orienting).
\end{proof}

\begin{figure}[t]
\centering
\includegraphics[width=0.40\textwidth]{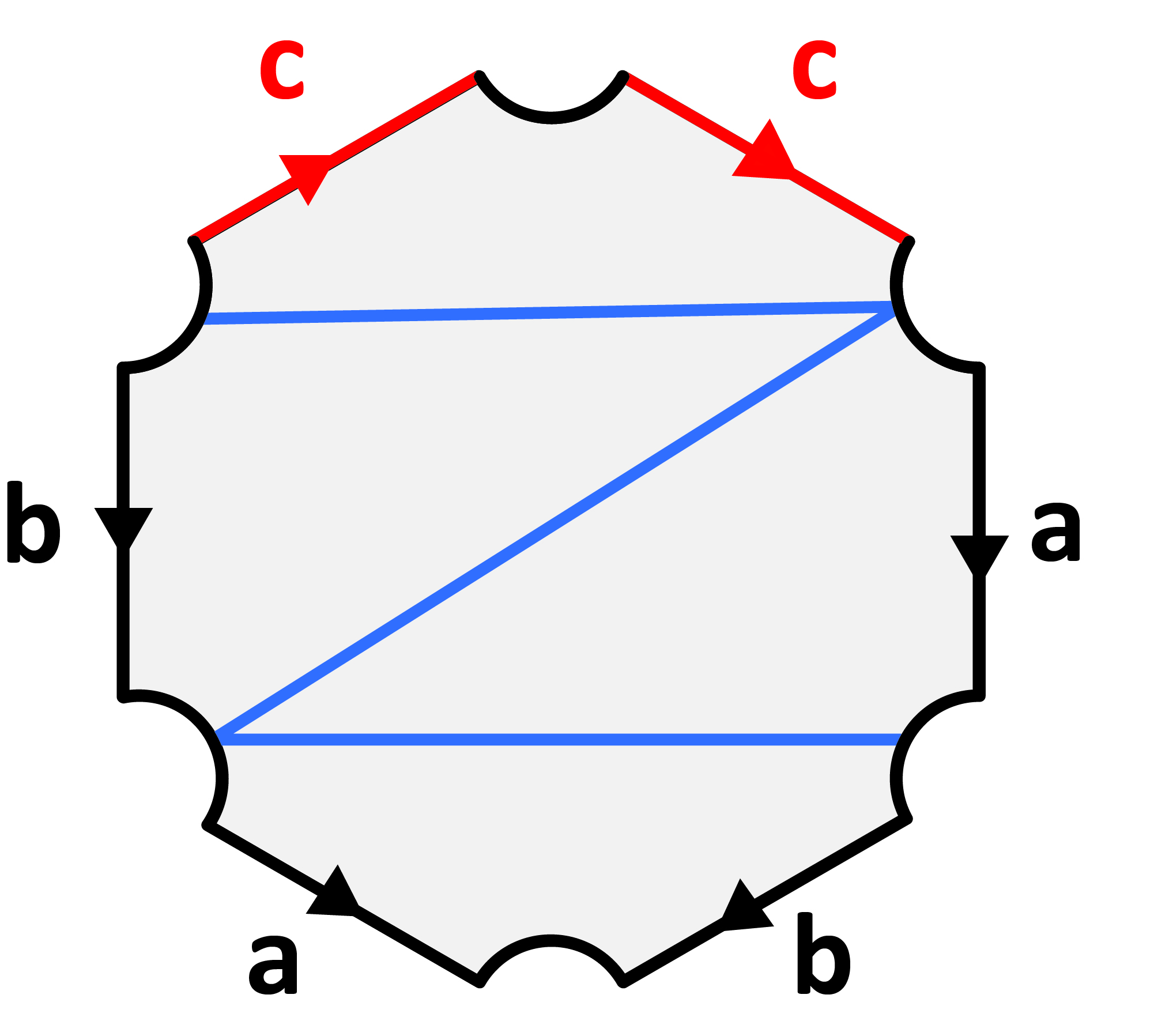} 
\caption{A non-orientable surface of genus 3 with embedded system of arcs. The dents in the picture indicate the segments of the boundary component.}
\label{counterexample}
\end{figure}

\subparagraph*{A Correction.} To prove Theorem \ref{T:negamimain}, we provide a different proof of Lemma~\ref{correction of negami}. The idea of the proof is to cut the surface along an orienting curve that does not cross a graph embedded on the surface too many times. By Lemma \ref{matouvsek}, we know that there exists such a curve. Then, we are on an orientable surface and can use Lemma~\ref{negami theorem}. Let us first introduce some additional terminology that is well tailored to surfaces with boundaries. 

Let $M$ be an orientable surface with boundaries that are given with orientations. Fix an orientation of $M$. We say that the orientations of the boundaries are mutually \emph{compatible} if $M$ is on the same side of each of them with respect to the chosen orientation. Figure~\ref{handle} shows a surface with non-compatible boundary orientations. For an arc with both ends on the same boundary component, we say the arc is two-sided (one-sided) if the closed curve obtained by connecting the two ends of the arc along one of the boundary segments is two-sided (one-sided). An orienting closed curve $\gamma$ is either two-sided or one-sided depending on the genus, as described in Lemma~\ref{matouvsek2}. In the same spirit of this lemma, one can characterize orienting arcs instead of orienting closed curves, the only difference is that when $g$ is odd, $h_\gamma=h$ and when $g$ is even, $h_\gamma=h+1$. This is because cutting along a two-sided arc increases the boundary components by one and cutting along a one-sided arc does not change the number of boundary components.

\begin{lemma}\label{compatible boundary orientation}
Let $N$ be a non-orientable surface of even genus and $M$ be the orientable surface obtained from cutting $N$ along an orienting curve. The orientations on the two boundaries of $M$ induced by the orienting curve  are compatible.
\end{lemma}

\begin{proof}
If the boundaries are not compatible, the surface we get by identifying the boundaries along their orientation is introducing a handle which contradicts the fact that $N$ was non-orientable. Figure \ref{handle} illustrates this.
\end{proof}
\begin{figure}[t]
\centering
\includegraphics[width=0.35\textwidth]{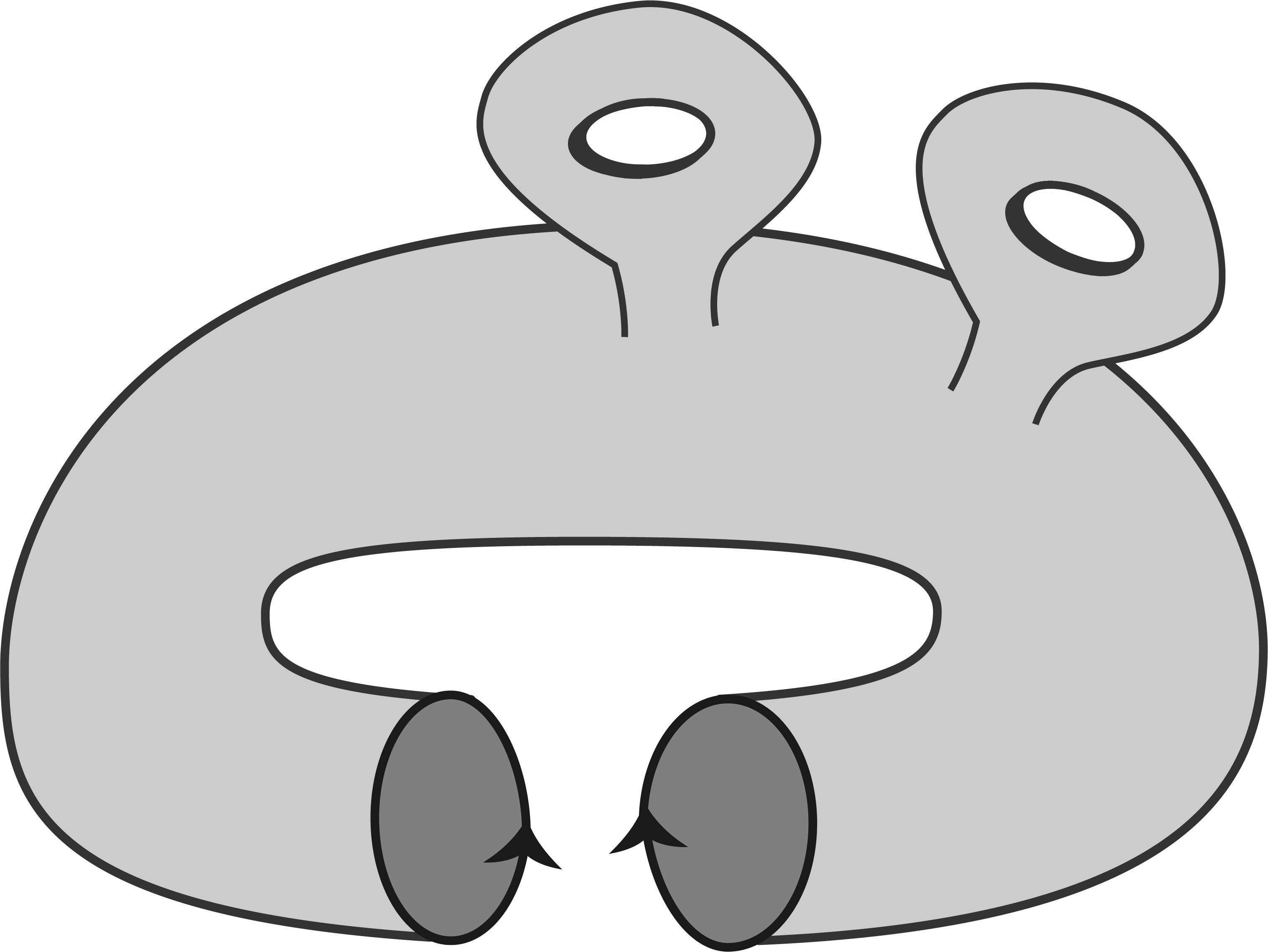} 
\caption{An orientable surface with non-compatible boundary orientations}
\label{handle}
\end{figure}

\begin{proof}[Proof of Lemma~\ref{correction of negami}]

For $i=1,2$, let $N_i$ be a non-orientable surface of genus $g$ and with one boundary component, with $\Gamma_i$, a system of $\beta_i$ disjoint essential proper arcs. We separate the cases where $g$ is odd from even.

{\large{\textbf{$\textbf{g}$ is odd}}.} 
Let $\gamma_i$ be an orienting arc in $N_i$ which has $c_i\leq 2\beta_i$ intersections with $\Gamma_i$, whose existence is guaranteed by Lemma~\ref{matouvsek} and Remark~\ref{remarkorienting}. Cut open $N_i$ along $\gamma_i$. Let $M_i$ be the resulting surface. From Lemma \ref{matouvsek2} we know that $M_i$ is an orientable surface of orientable genus $\frac{g-1}{2}$ and a boundary component. Each arc in $\Gamma_i$ is cut into at most $3$ arcs by $\gamma_i$. 
We denote by $\Gamma_i'$ the system of disjoint essential arcs in $M_i$ and thus we have $|\Gamma_i'|\leq 3\beta_i$.
From Lemma \ref{negami theorem}, we know there is a homeomorphism $\phi_i$ which maps $M_i$ to the surface $M^{'}$ with one boundary component, such that $\phi_1(\Gamma_1')$ and $\phi_2(\Gamma_2')$ have at most $4\frac{g-3}{2}|\Gamma_1'||\Gamma_2'|\leq 36\frac{g-3}{2}\beta_1 \beta_2$ crossings.

\begin{figure}[t]
\centering
\includegraphics[width=0.4\textwidth]{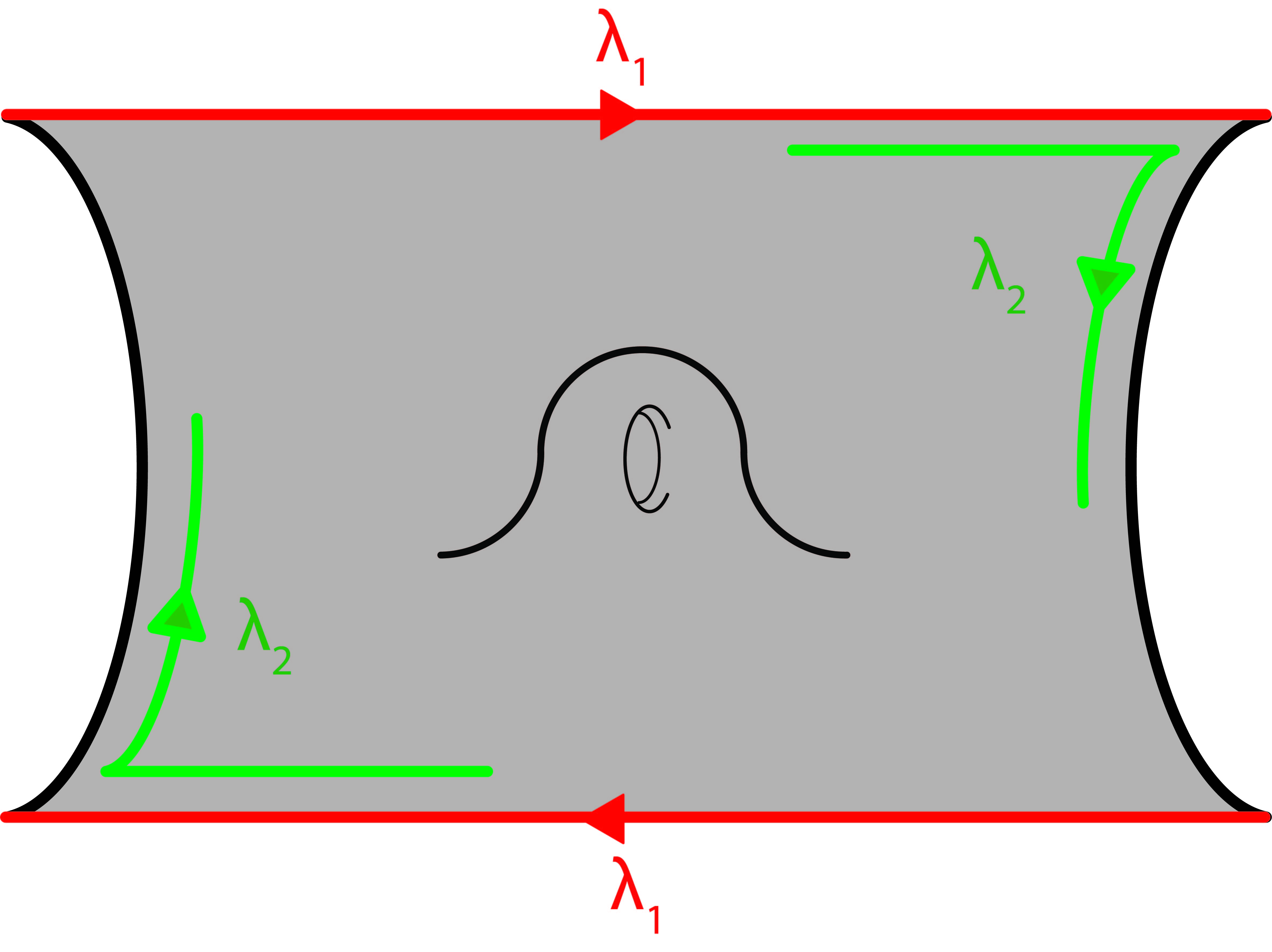} 
\caption{The non-orientable surface cut open along an orienting curve.}
\label{g odd}
\end{figure}

We modify $\phi_2$ in a small neighborhood of the boundary so that the copies of $\gamma_2$ are aligned with those of $\gamma_1$. This will allow us to glue back the surface along $\gamma_1$ (or $\gamma_2$). In order to do so, we slide out the ends of the arcs in $\Gamma_2'$ which are not on $\gamma_2$, containing $2\beta_2$ ends, into the two segments of the boundary of $N_1$. 
Each one of the ends we are sliding might intersect each end of the arcs in $\Gamma_1'$ which lies in the two segments of the primary boundary component which contains $2\beta_1$ ends.  Therefore this modification introduce at most $4\beta_1\beta_2$ intersections. Each copy of $\gamma_i$ contains $c_i$ ends of the arcs $\Gamma_i'$. Next we align copies of $\gamma_1$ with $\gamma_2$ and this introduces at most $c_1c_2$ new intersections on each copy. Therefore we will have $4\beta_1\beta_2+2c_1c_2\leq 12\beta_1\beta_2$ new intersections.

After this modification, we are able to glue back $M^{'}$ along $\gamma_1$ (or $\gamma_2$) to get back to a non-orientable surface $F$ with genus $g$ and a boundary component.
Now $\phi_1$ and $\phi_2$ introduce two homeomorphism from $N_1$ and $N_2$ into $N$ such that $\phi_1(\Gamma_1)$ and $\phi_2(\Gamma_2)$ intersect at most $12\beta_1\beta_2+ 36\frac{g-3}{2}\beta_1\beta_2\leq 18(g-1)\beta_1\beta_2$.

{\large{\textbf{$\textbf{g}$ is even}}.} Let $\gamma_i$ be an orienting arc in $N_i$ which has $c_i\leq 2\beta_i$ intersections with $\Gamma_i$ whose existence is guaranteed by Lemma~\ref{matouvsek} and Remark~\ref{remarkorienting}. Let $M_i$ be the surface obtained by cutting $N_i$ along $\gamma_i$ and denote by $\Gamma_i'$ the  system of arcs in $M_i$. Each arc in $\Gamma_i$ is cut into at most $3$ arcs by $\gamma_i$. 
We denote by $\Gamma_i'$ the system of disjoint essential arcs in $M_i$ and we have $|\Gamma_i'|\leq 3\beta_i$. From Lemma \ref{matouvsek2} we know that $M_i$ is an orientable surface of genus $\frac{g-2}{2}$ and two boundary components with $3\beta_i$ disjoint essential arcs. Fixing an orientation for $\gamma_i$ will let us see how the two boundary components were pasted and according to Lemma \ref{compatible boundary orientation}, we know that these orientations are compatible (see Figure \ref{g even N1}). 
\begin{figure}[t]
\centering
\includegraphics[width=0.37\textwidth]{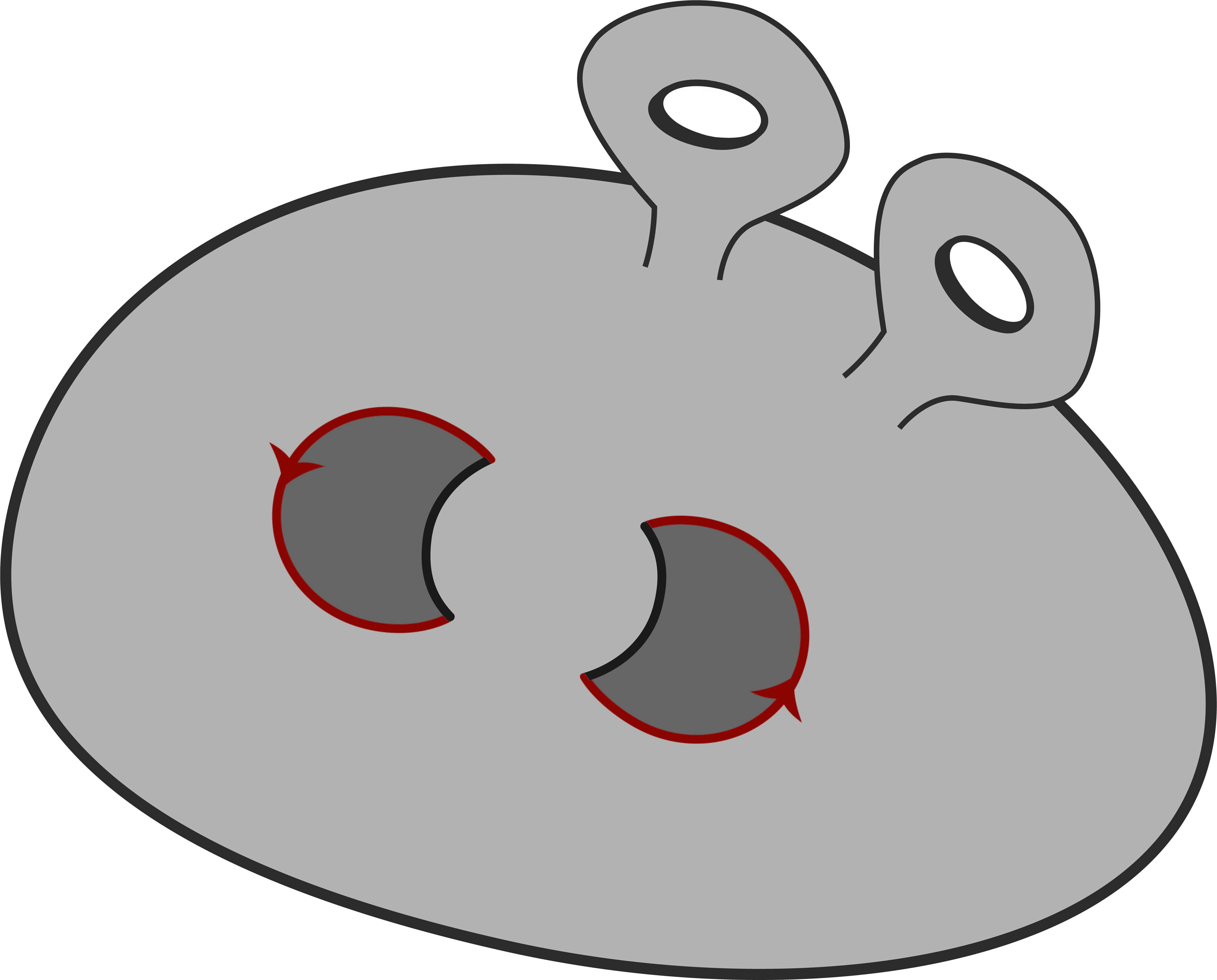} 
\caption{The surface $M_i$, the red arcs are the copies of $\gamma_i$}
\label{g even N1}
\end{figure}

 We claim that there is a curve $\nu_i$ with an end on each of the boundary components and with at most $3\beta_i$ intersections with the system of arcs in $M_i$. We know that cutting along an arc increases the Euler characteristic of the surface by 1 and we chose $\nu_i$ such that it reduces the boundary by 1. From the relation $\chi=2-2g-h$ where $h$ is the number of boundary components, we will see that after the cut, the genus is unchanged.
Therefore, by cutting along $\nu_i$, we get the orientable surface $M_i^{'}$ with genus $\frac{g-2}{2}$ and one boundary component with $6\beta_i$ disjoint essential arcs.

We choose $\nu_i$ such that its ends lie on the segments on the boundary components which belonged to the primary boundary component of $N_i$ and such that it connects these two boundaries. The problem reduces to finding such an arc so that it has at most $3\beta_i$ intersections with the system of arcs in $M_i$.
Consider the graph $H$ such that the vertices are the end points of the arcs and the edges are the arcs $\Gamma_i'$ which we denote by $E_1$ together with the segments on the boundary components denoted by $E_2$. We choose a face $\rho$ which has one edge in $E_2$ such that a part of this edge lies on the segment of the primary boundary component. We choose the face $\rho'$ analogously to $\rho$ but on the other boundary component. The shortest path between vertices associated to $\rho$ and $\rho'$ in the dual graph of $H$ induces a path from $\rho$ to $\rho'$ which after connecting the ends to the edge on the boundary component in both $\rho$ and $\rho'$, we will have an arc $\nu_i$ which passes through the inner faces at most once (because of our choice of shortest path in the dual graph) and therefore intersects each edge of $E_1$ at most once. Thus $\nu_i$ is the desired arc. 

\begin{figure}[t]
\centering
\includegraphics[width=0.43\textwidth]{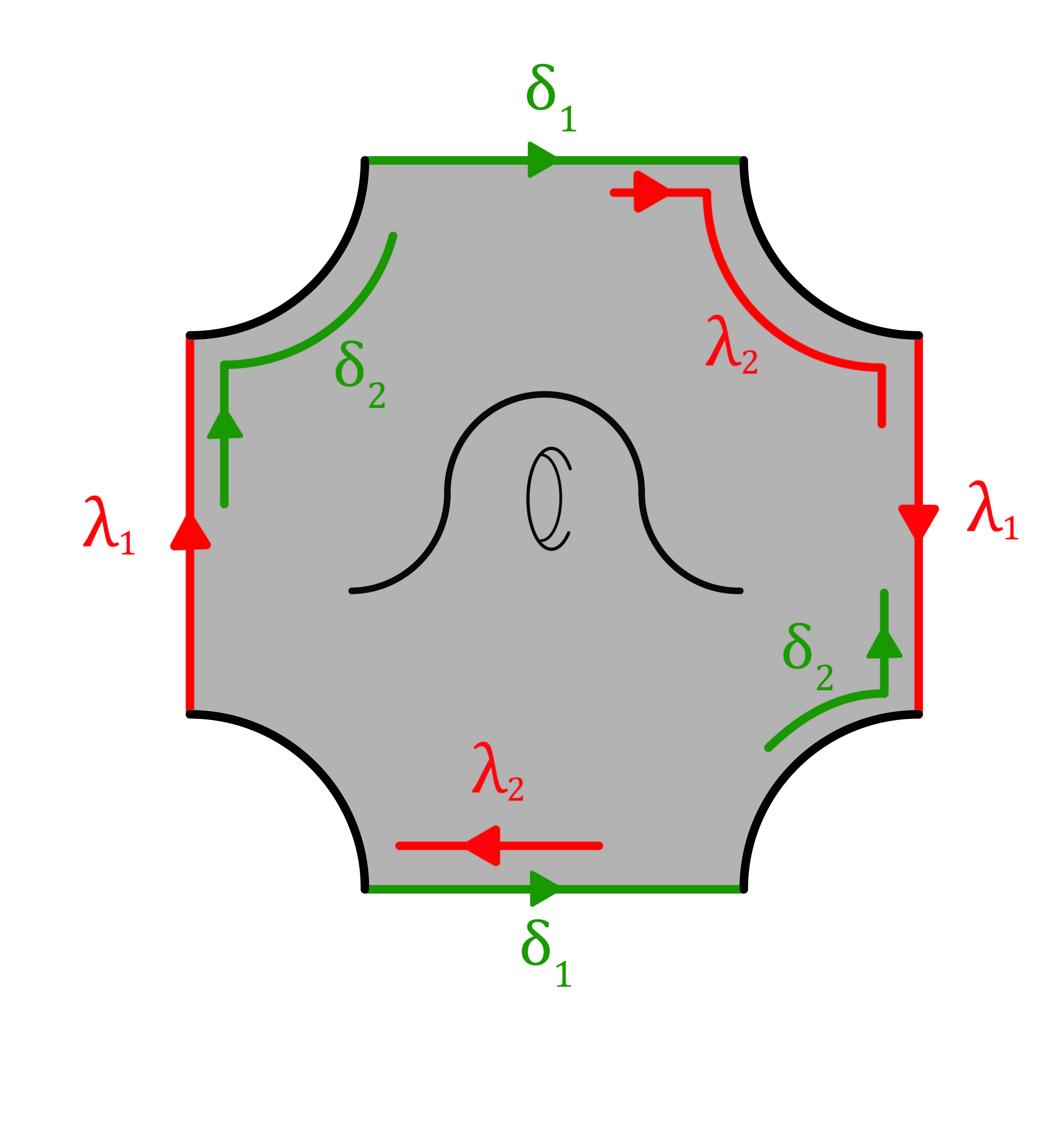} 
\caption{The orientable surface $M^{'}$ with one boundary component.} 
\label{g even N2}
\end{figure}

Again, from Lemma \ref{negami theorem} we know there exists $\phi_1$ and $\phi_2$ mapping $M_1^{'}$ and $M_2^{'}$ to $M^{'}$ with \mbox{$4\frac{g-4}{2}\cdot6\beta_1\cdot6\beta_2$} intersections in the arcs (see Figure \ref{g even N2}). Similar to the case where the genus was odd, we can modify $\phi_1$ and $\phi_2$ by introducing at most $4\cdot 6\beta_1\cdot 6\beta_2$ intersections such that we can align copies of $\gamma_1$ with $\gamma_2$ and $\nu_1$  with $\nu_2$ and glue back $M^{'}$ to obtain $N$ and at the end we have at most $72(g-2)\beta_1\beta_2$. This finishes the proof.
\end{proof}

\begin{proof}[Sketch of proof of Theorem~\ref{T:negamimain}]
Once we are equipped with Lemma~\ref{correction of negami}, the proof of Theorem~\ref{T:negamimain} follows the same strategy as that of Negami~\cite{negami2001crossing}. For completeness, we provide the following sketch. For each $i=1,2$, we contract a spanning tree in $G_i$, reducing it to a one-vertex embedding scheme, and remove contractible arcs, yielding a scheme $S_i$. We puncture the surface at the single vertex of $S_i$, yielding a non-orientable surface $N_i$ with one boundary component and $O(|E(G_i)|)$ disjoint essential arcs. We are now in the situation to apply Lemma~\ref{correction of negami}, which gives us a pair of homeomorphisms sending $N_1$ and $N_2$ to $N$ such that the number of crossings between the arcs is $O(g|E(G_1)||E(G_2)|)$. We glue back a disk on the puncture and connect all the incoming arcs of $S_i$ to a base-point $b_i$, so that the two base points $b_1$ and $b_2$ are slightly off. This induces an additional $O(|E(G_1)||E(G_2)))$ crossings. Finally, we add back the contractible arcs and uncontract the spanning trees in small neighborhoods of $b_1$ and $b_2$, which does not change the number of crossings, concluding the proof.

\end{proof}

\section{Non-orientable Canonical System of Loops}\label{non-canonical}

The objective of this section is to prove the following theorem.

\canonical*

In order to prove this theorem, we heavily rely on an algorithm introduced by Schaefer and \v{S}tefankovi\v{c} in~\cite{JGAA-580} for drawing graphs on non-orientable surfaces. Compared to that algorithm, our approach enforces more structure on this construction by imposing specific orders in the way we deal with the loops in the induction.  In Section~\ref{schaefer algorithm}, we first explain this algorithm and then in Section~\ref{modification}, we introduce and explain our modifications, ultimately leading to a proof of the theorem in Section~\ref{the system}.

\subsection{The Schaefer-{\v{S}}tefankovi{\v{c}} Algorithm}\label{schaefer algorithm}

Throughout this section, we will be working with a one-vertex graph $G$ endowed with an embedding scheme $(\rho,\lambda)$, which for simplicity we denote by $G$.

Schaefer and {\v{S}}tefankovi{\v{c}} proved the following theorem.

\begin{theorem}\label{schaefer lemma}\cite[Lemma~9]{JGAA-580}
If $G$ is a one-vertex non-orientable (respectively orientable) scheme $(\rho,\lambda)$, then it admits a cross-cap drawing with $eg(G)$ (respectively $eg(G) + 1$) cross-caps in which every edge passes through every cross-cap at most twice.
\end{theorem}

The proof is by induction on the number of loops and is readily algorithmic, computing an actual cross-cap drawing with the required bound on the number of intersections with the cross-caps. Before explaining the algorithm, let us introduce the different moves and techniques we use to deal with different types of loops. 

First, let us introduce the following terminology for an embedding scheme around a vertex $v$. By \emph{flipping} a wedge of a one-sided loop $e$ in a one-vertex scheme, we mean reversing the order of the edges in the wedge and changing the signature of the loops that have exactly one end in the wedge. We call the empty wedge between two consecutive half-edges around $v$ a \emph{root wedge}. Note that a face incident to $v$ in a cross-cap drawing may correspond to more than one root wedge. We refer to such a face as a \emph{root face}.

\textbf{Contractible loop move}. Let $c$ be a contractible loop with consecutive ends in the scheme $G$. Remove $c$. The new scheme can be drawn using the same number of cross-caps. Having a drawing for the new scheme, we can draw the loop $c$ without passing through any of the cross-caps.

\textbf{Gluing move}. Let $s$ be a non-contractible separating loop in the scheme $G$. We divide the scheme to $G_1$ and $G_2$ by cutting along $s$ and splitting the vertex into two vertices (the embedding schemes of $G_1$ and $G_2$ are induced by the embedding scheme of $G$). Denote by $f_{o}^1$ and $f_{o}^2$ the root wedges in $G_1$ and $G_2$ in which $s$ formerly was placed. 
Let $H_1$ and $H_2$ be drawings for $G_1$ and $G_2$ respectively. We glue the drawings by identifying the root wedges $f_{o}^1$ and $f_{o}^2$ to get the drawing $H^{'}$ for $G\setminus \{s\}$. 

 Note that removing $s$ does not change the genus (Lemma \ref{adding separating}) and we have $eg(G)=eg(G_1)+eg(G_2)$. If $G_1$ and $G_2$ are both non-orientable, then $H^{'}$ can be extended to a cross-cap drawing for $G$ by adding $s$ without using any of the cross-caps; see Figure \ref{glue}.
When at least one of $
G_1$ or $G_2$ is orientable, say $G_2$, $H^{'}$ uses one extra cross-cap ($G_2$ needs $eg(G_2)+1$ cross-caps to be drawn). In order to get a drawing with minimum number of cross-caps, we need to eliminate one cross-cap from the drawing. To deal with this case, we need the following lemma and the dragging move which allows us to reduce one cross-cap from the drawing.

\begin{lemma} \label{adding orienting}
Let $(\rho,\lambda)$ be an orientable embedding scheme for the one vertex graph $G$. Adding a one-sided loop $o$ with consecutive ends to the scheme (anywhere in the rotation around the vertex) increases the Euler genus by 1. Thus, the new scheme needs as many cross-caps as $G$ to embed. Furthermore, the loop $o$ is orienting.
\end{lemma}
\begin{proof}
Adding a one-sided curve with consecutive ends to an embedding scheme does not change the number of faces. Then, by the Euler formula we know that the Euler genus of the new embedding scheme is increased by 1. By lemma \ref{odd non-orientable genus}, $(G,\rho,\lambda)$ needs $eg(G)+1$ cross-caps to embed. Therefore, the new scheme needs as many cross-caps as $G$ to embed. \\

The loop $o$ is the only one-sided loop in the embedding scheme and every other loop is in the same wedge of $o$. Lemma \ref{orienting loops} implies that $o$ is orienting.
\end{proof}

 \begin{figure}[t]
    \centering
    \includegraphics[width=.7\textwidth]{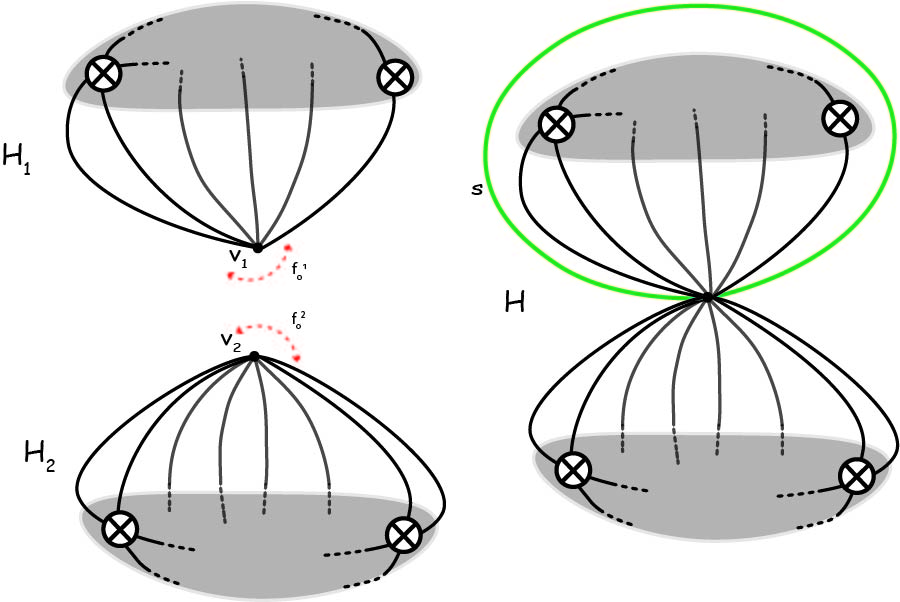}
     \caption{The gluing move on two cross-cap drawings when $G_1$ and $G_2$ are both non-orientable.}
    \label{glue}
\end{figure}

\textbf{Dragging move}. Let us assume that $G_2$ is orientable. By Lemma \ref{adding orienting}, we can add a one-sided loop $o$ with consecutive ends in the root wedge $f_{o}^2$ without increasing the number of cross-caps that we need to draw $G_2$. The loop $o$ is orienting and the new scheme needs $eg(G_2)+1$ cross-caps to be drawn. Having a drawing for the new scheme, we can draw the loop $s$ in the drawing for $G_2+\{o\}$ as follows: we start the loop from one of the root wedges between $o$ and another loop of $G_2$, we draw $s$ by following $o$ through all the cross-caps, except that after coming out of the last cross-cap, we go back to the first one entered, and traverse all of the cross-caps again. At the end, we follow $o$ back to the vertex; see Figure \ref{glue and drag}. We denote this drawing of $G_2+\{o\}+\{s\}$ by $H^{'}_2$. 

By gluing $H_1$ to $H_2^{'}$, we get a drawing $H^{'}$ for $G+\{o\}+\{s\}$ but the drawing is not using the minimum number of cross-caps. We eliminate one of the cross-caps in $H^{'}$ as follows.

Let $i_1$ be the rightmost half-edge in $G_1$ that follows immediately the separating loop in $G$. Denote by $\mathfrak{c}$ the first cross-cap that $i_1$ passes through. Let us assume that there are $2k$ half-edges passing through $\mathfrak{c}$. Let us denote by $(i_1,f_1,\dots, i_{2k}, f_{o}^1)$ the alternating sequence of half-edges and faces adjacent to $\mathfrak{c}$ in the cross-cap drawing by moving clockwise around it. 
Now, we disconnect the edges that enter $\mathfrak{c}$ and remove the cross-cap $\mathfrak{c}$. We drag $i_1,\dots, i_k$ through all the cross-caps in $G_2$ along the loop $o$. After exiting the last cross-cap in $G_2$, we remove $o$ and we attach the half edges to their other ends ($i_{k+1},\dots, i_{2k}$). Since $G_2$ uses an odd number of cross-caps (Lemma \ref{odd non-orientable genus}), the half edges will have the  correct orientability and order to get attached to their other ends;
see Figure \ref{glue and drag}. 
If only one of $G_1$ and $G_2$ is orientable, the drawing we get uses $eg(G)$ cross-caps and if both are orientable, we get a drawing with $eg(G)+1$ cross-cap that is the minimum cross-cap needed to draw the scheme in this case.

 \begin{figure}[t]
    \centering
    \includegraphics[width=\textwidth]{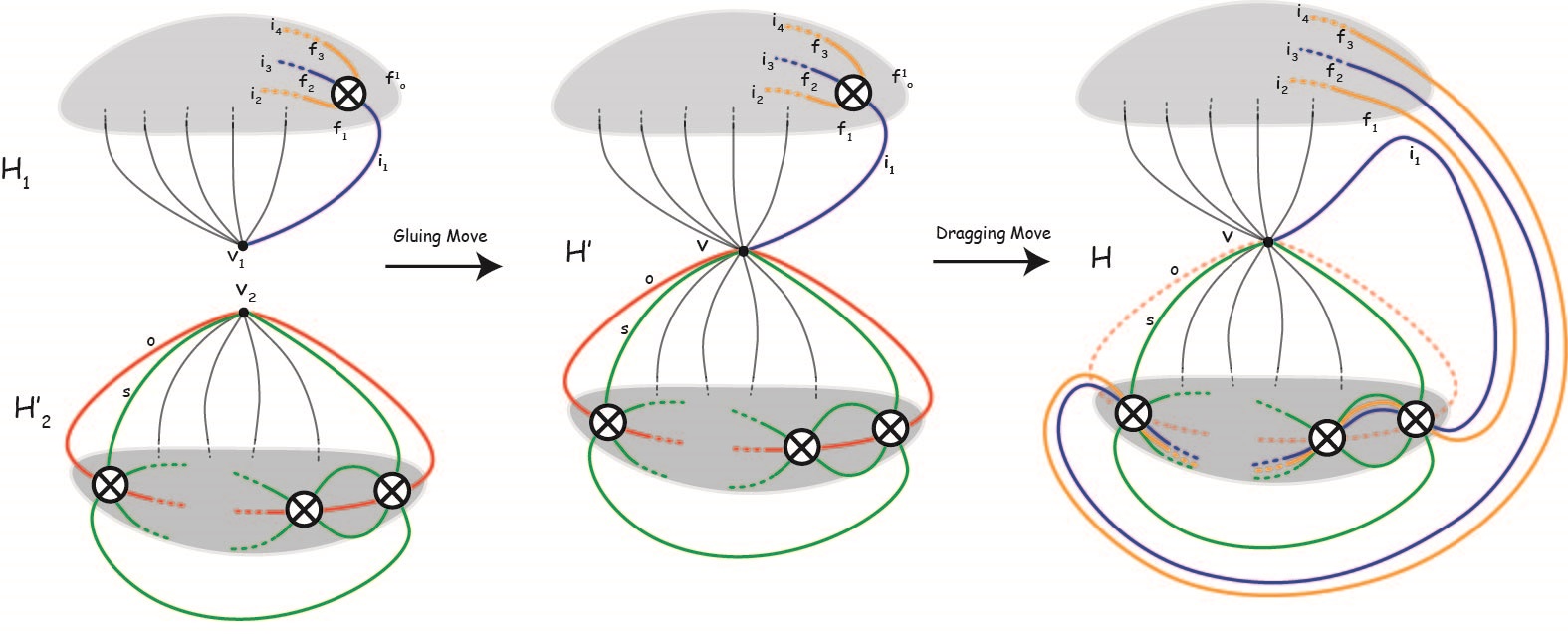}
     \caption{Left: the gluing move. Right: the dragging move when $G_2$ is orientable: the top right crosscap is removed and the corresponding curves are dragged through the bottom component.}
    \label{glue and drag}
\end{figure}

\textbf{One-sided loop move}. Let $r$ be a one-sided loop in the scheme $G$. We remove $r$ and flip one of its wedges. One can check that the new scheme $G^{'}$ has Euler genus $eg(G)-1$. Let us assume that $H^{'}$ is a drawing for $G^{'}$. We add $r$ to this drawing by adding a cross-cap near the vertex and the flipped wedge and dragging $r$ and every edge in the flipped wedge in it; see Figure \ref{one-sided move}.
Note that flipping different wedges of $r$ leads to two different cross-cap drawings. This freedom in choosing the wedge is important for us and we use this later in the paper.

\begin{figure}[t]
    \centering
    \includegraphics[width=.65\textwidth]{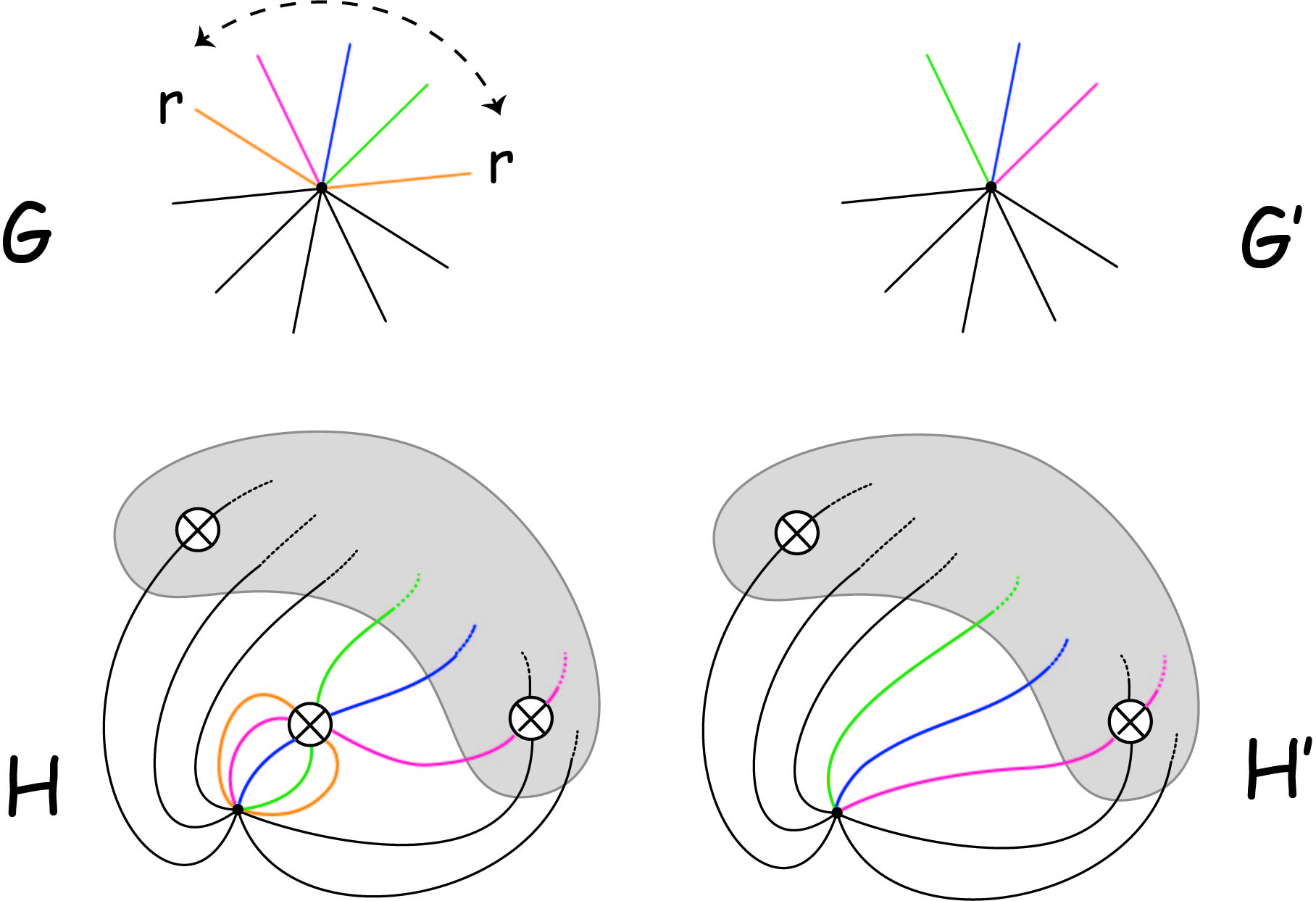}
     \caption{The one-sided loop move on the loop $r$.}
    \label{one-sided move}
\end{figure}

If $r$ is not orienting, the drawing we get at the end uses $eg(G)$ cross-caps.
But if $r$ is orienting, then $G^{'}$ is orientable and any drawing for $G^{'}$ needs an extra cross-cap (Lemma \ref{odd non-orientable genus}). This means that if we apply a one-sided loop move on an orienting loop, the drawing we get does not use the minimum amount of cross-caps (the embedding is not cellular).

We use the following move to deal with orienting loops.

\textbf{Concatenation move.} Let $o$ be an orienting loop in the scheme $G$ such that one of its ends is immediately followed by an end of a two-sided non-separating loop $t$ in the rotation. By Lemma~\ref{orienting loops}, since $t$ is non-separating, the concatenation of $o$ and $t$ which we denote by $o^{'}$, is not orienting. Denote by $G^{'}$ the scheme in which we replace $o$ by $o^{'}$ (we need $eg(G)$ cross-caps to draw both $G$ and $G^{'}$). If $H^{'}$ is a drawing for $G^{'}$, one can obtain from $H^{'}$ a drawing of $G$ by replacing the drawing $o^{'}$ by its concatenation with $t$. Depending on the wedge of $o^{'}$ that we choose to flip, we slide $o^{'}$ along $t$ in the drawing:

If we flipped the wedge that does not encompass the loop $t$, we detach the end of $o^{'}$ next to $t$ and slide it along $t$ and we attach it to the vertex. This way, it ends up where the end of $o$ was placed originally; see Figure \ref{concatenate1}.
\begin{figure}[t]
    \centering
    \includegraphics[width=\textwidth]{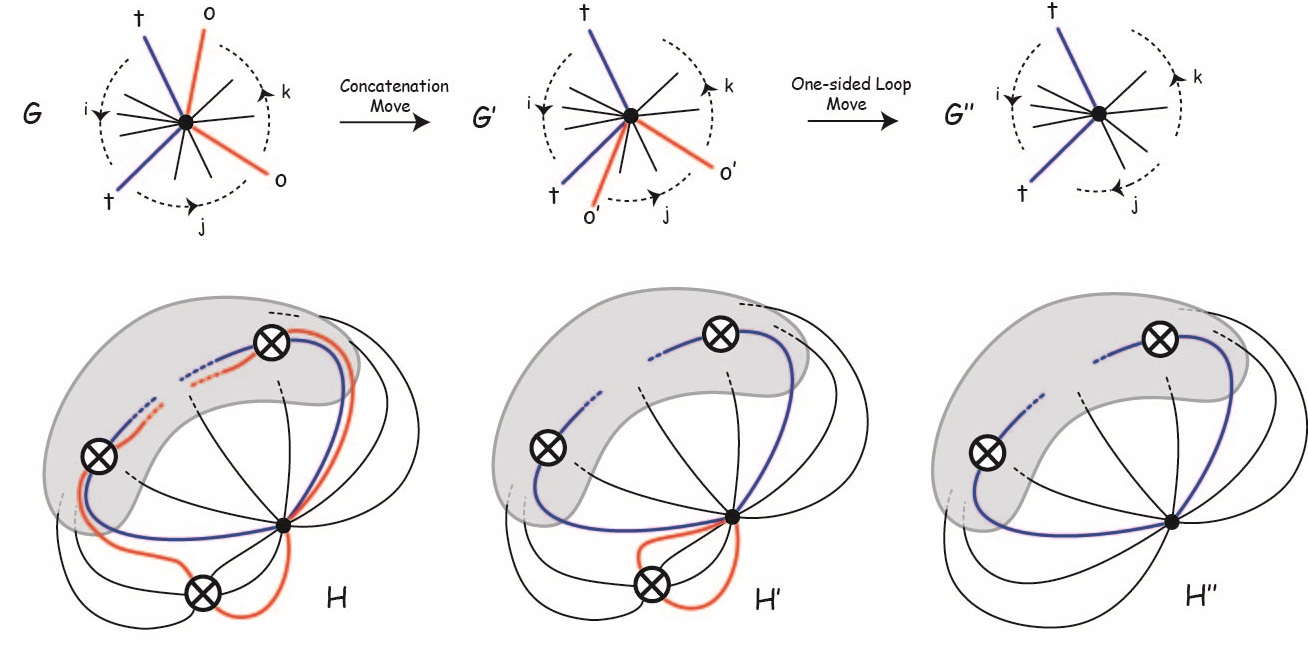}
    \caption{The concatenation move on the loops $o$ and $t$, when in applying the one-sided loop move we flip the wedge of $o^{'}$ that does not encompass the ends of the loop $t$.}
    \label{concatenate1}
\end{figure}If we flipped the wedge that encompasses the loop $t$, we draw $o$ as follows: note that $o^{'}$ passes trough only one cross-cap. We draw $o$ next to the end of $o^{'}$ that is not slid along $t$, but instead of following $o^{'}$ into the cross-cap, we follow $t$. We can do this because the loop $o^{'}$ is next to the loop $t$ in the rotation around this cross-cap; see Figure \ref{concatenate2}.

 \begin{figure}[t]
    \centering
    \includegraphics[width=\textwidth]{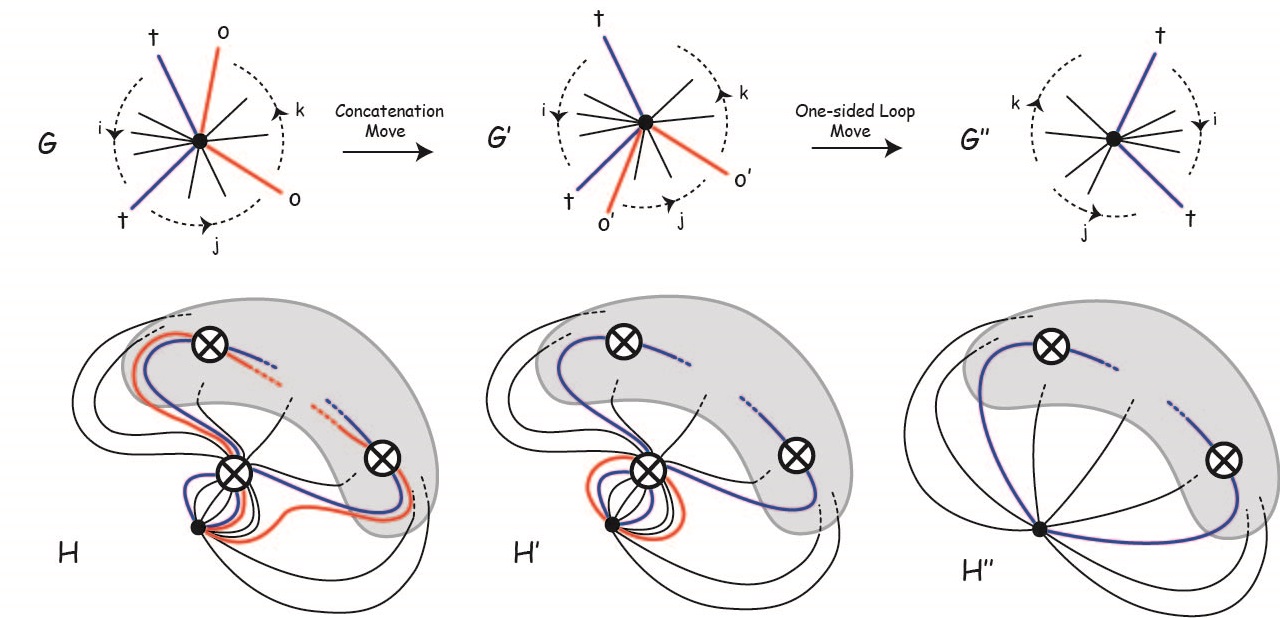}
    \caption{The concatenation move on the loops $o$ and $t$, when in applying the one-sided loop move we flip the wedge of $o^{'}$ that encompass the ends of the loop $t$.}
    \label{concatenate2}
\end{figure}
\textbf{Exchange move.} Let $G$ be a scheme that has only two-sided loops and  no separating loop. We exchange two consecutive half-edges and change the signatures of the corresponding edges. One can prove that the new scheme can be drawn using the same number of cross-caps as the initial scheme. Having a drawing of the new scheme, we obtain a drawing of the original scheme by adding a cross-cap near the reversed half-edges. This drawing is using $eg(G)+1$ cross-caps and this is the minimum number of cross-caps we need to draw an orientable scheme (Lemma \ref{odd non-orientable genus}).

Each of these moves provides a way to draw a loop assuming that some simpler one-vertex graph without that loop has already been drawn. Therefore, we can use the moves in an inductive algorithm as follows. The input is an embedding scheme $G$.

\begin{tcolorbox}

\textbf{The Schaefer-\v{S}tefankovi\v{c} algorithm:}
\begin{itemize}

\item \textbf{Step 1: If there exists a contractible loop.} We recurse on the scheme without the loop and apply the contractible loop move.
 \item \textbf{Step 2. If there exists a separating (non-contractible) loop $s$.} We divide the scheme into two sub-schemes on each side of $s$. We have the following cases:
 \begin{itemize}
   \item \textbf{Step 2.1: Both sub-schemes are non-orientable.} We recurse on the two sub-schemes and apply the gluing move.
   \item \textbf{Step 2.2: At least one of the sub-schemes is orientable.}  We recurse on the two sub-schemes and apply the gluing move followed by the dragging move.
\end{itemize}
  \item  \textbf{Step 3.1: If there is a one-sided non-orienting loop $r$.} We recurse on the scheme without the loop $r$ and the flipped wedge, and apply the one-sided loop move on the loop $r$.
  \item  \textbf{Step 3.2: If all one-sided loops are orienting.} If there exists no two-sided loop, by Lemma \ref{orienting loops}, all pairs of loops are interleaving and we can draw the scheme with one cross-cap so that each loop enters it once. If there exists a two-sided loop in the scheme, we can find a place in the scheme where an orienting loop $o$ is followed immediately by a two-sided loop $t$. We recurse on the scheme $H^{'}$ described in the concatenation move, and apply the concatenation move on the curve $o$.
  \item \textbf{Step 3.3: If every loop in the scheme is two-sided.} We apply the exchange move on two consecutive half-edges and recurse on the new scheme.  

\end{itemize}
\end{tcolorbox}

The proof of Schaefer and \v{S}tefankovi\v{c} of Theorem~\ref{schaefer lemma} proceeds by analyzing this algorithm and proving that (1) the resulting cross-cap drawing has the correct number of cross-caps and (2) each loop passes through every cross-cap at most twice.

\begin{remark}\label{enteringonce}
By Lemma \ref{homology}, in a drawing that is obtained by this algorithm, every orienting loop passes through each cross-cap exactly once and if a separating loop enters a cross-cap, it passes through that cross-cap exactly twice.
\end{remark}

There is some leeway in this algorithm: while the steps have to be applied in this specific order, in each step a loop of the given type is chosen arbitrarily. Our modification of the algorithm follows the exact same blueprint but enforces specific orders in which we choose separating and one-sided non-orienting loops. These specific orders provide more structure to the resulting drawing, make it lend itself more to connecting the cross-caps, in order to form the desired non-orientable canonical system of loops of low multiplicity.

\subsection{Our Modification to the Schaefer-{\v{S}}tefankovi{\v{c}} algorithm}\label{modification}

Our algorithm first starts with some preprocessing.

\subsubsection{Preprocessing}

\begin{lemma}\label{reducing}
Given a graph $G$ embedded on a non-orientable surface $N$, there exists a one-vertex scheme $\hat{G}$ such that $\hat{G}$ has an orienting loop, and if $\hat{G}$ has a non-orientable canonical system of loops such that each loop crosses each edge of $\hat{G}$ at most $k$ times, then $G$ has a non-orientable canonical system of loops such that each loop crosses each edge of $\hat{G}$ at most $3k$ times.
\end{lemma}

\begin{proof}
 By Lemma \ref{matouvsek}, there exists an orienting curve $\gamma$ embedded on the surface $N$ that has multiplicity two. Denote by $G^{'}$ the overlay of $G$ and $\gamma$. Note that each edge in $G$ is divided into at most three sub-edges in $G^{'}$. If $\gamma$ crosses $G$, $n$ times totally, then it corresponds to $n$ edges in $G^{'}$. Choose a spanning tree $T$ in $G^{'}$ that contains $n-1$ of these edges and contract the edges of $T$ into a single point to get a one-vertex graph $\hat{G}$. By merging the rotations, we get an embedding scheme for the one-vertex graph $\hat{G}$. Note that the non-contracted sub-edge of $\gamma$ is orienting in $\hat{G}$.
 
 Let us assume that $\hat{G}$ has a non-orientable canonical system of loops such that each loop crosses each edge of $\hat{G}$ at most $k$ times. We can uncontract the spanning tree $T$ close to the vertex so that the uncontracted edges do not cross the canonical system of loops. This gives us a canonical system of loops for $G^{'}$. To get a drawing for $G$, we remove the orienting curve $\gamma$. Since any edge of $G$ is formed by at most three edges of $G^{'}$, then the canonical system of loops for $G^{'}$ crosses each edge of $G$ at most $3k$ times
\end{proof}
\begin{remark}\label{coefficient3}
Note that the curve $\gamma$ divides each edge of $G$ into at most three sub-edges. Therefore, each edge in $G$ is obtained as the concatenation of at most 3 edges of $\hat{G}$.
\end{remark}

 This lemma lets us assume, at the cost of a slightly bigger multiplicity, that (1) the input graph is a one-vertex graph endowed with an embedding scheme, and (2)  the embedding scheme we work with contains an orienting loop. This orienting loop will in turn allow us to avoid some steps in the algorithm.

For the second prepossessing move we need a definition that is inspired by a similar notion from the literature on sorting signed permutations by reversals~\cite{hannenhalli1999transforming}. Given an embedding scheme $G$, the \emph{interleaving graph} $I_G$ has as vertex set the set of loops of $G$, and two vertices are connected if their corresponding loops have interleaving ends, see Figure \ref{one-sided moves} for an example. When we talk about the sidedness of a vertex, we mean the sidedness of the loop it is associated to. A connected component in the interleaving graph is called \emph{non-orientable} if it has a one-sided vertex, and \emph{orientable} otherwise. We call a component with only one vertex a \emph{trivial} component, and \emph{non-trivial} otherwise. Separating loops (contractible or non-contractible) correspond to isolated two-sided vertices, i.e., trivial orientable components, in the interleaving graph.

Our second preprocessing step aims at subdividing $G$ into subschemes $G_i$ such that each $I_{G_i}$ has only one non-trivial component. In order to do this, we \emph{saturate} the scheme with auxiliary separating loops, i.e., we add a separating loop for any non-trivial component that is not divided from the rest of the graph by some separating loops.

\begin{remark}
By Lemma \ref{adding separating}, we can see if $\bar{G}$ is the scheme we obtain after saturating $G$, then $eg(G)=eg(\bar{G})$. Thus, after drawing $\bar{G}$, we can remove these auxiliary separating loops to get a drawing for $G$.
\end{remark}

Given a non-orientable scheme $G$ saturated with separating loops and any cellular embedding of $G$ on a surface $N$, cutting $G$ along the separating loops yields subsurfaces $N_i$ of $N$, each containing (possibly empty) components of $G$, which we denote by $G_i$ (see Figure~\ref{component tree}). The \emph{component tree} of $G$ has a  vertex for every such sub-graph $G_i$, and two vertices are connected if their corresponding components are separated by a separating loop. See Figure \ref{component tree} for an example of a component tree. We shall quickly that in the context of our algorithm there will actually be exactly one  non-orientable component. We root the tree on the vertex corresponding to the non-orientable component.

 \begin{figure}[t]
    \centering
    \includegraphics[width=0.75\textwidth]{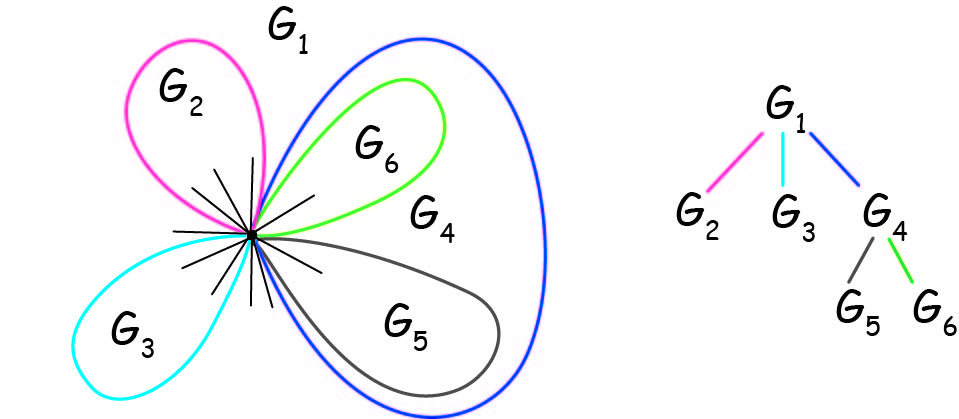}
     \caption{Left: a saturated one-vertex scheme in which the drawn loops are the separating loops; right:the component tree of the left scheme. Note that the component $G_4$ is an empty sub-scheme.}
    \label{component tree}
\end{figure}

In a saturated scheme, each separating loop relates two sub-graphs $G_i$ that exist on its different sides. If we have $k$ non-trivial components or empty subgraphs in total in the interleaving graph, in order to separate all of them from each other, we need exactly $k-1$ separating loops.

\subsubsection{A cross-cap drawing algorithm}

We are now ready to describe our modified algorithm. Our modification consists of forcing the presence of an orienting curve on the graph using Lemma~\ref{reducing} and putting more structure on the order we deal with different separating loops and different one-sided non-orienting loops.\\

Throughout the main loop of our algorithm, if we have homotopic loops we remove all of them except one and after drawing this one, we re-introduce them parallel to the one drawn.

\begin{tcolorbox}
\textbf{The modified algorithm}:

\textbf{Pre-processing steps:}
\begin{itemize}
    \item \textbf{Step A.} If there is no orienting loop, we add an orienting loop and contract a spanning tree using Lemma \ref{reducing}. 
    
 
    \item \textbf{Step B.} If $G$ is not saturated by separating curves, we saturate it.
\end{itemize}
\textbf{Main loop:}
\begin{itemize}
    \item \textbf{Step 1: If there is a contractible loop.} We recurse on the scheme without the loop and apply the contractible loop move.
    
    \item \textbf{Step 2: If there exists a separating (non-contractible) loop.} We pick a separating loop that separates a non-root leaf from the component tree, recurse on the subschemes and apply a gluing and a dragging move.

    \item \textbf{Step 3.1: If there exists a one-sided non-orienting loop.} We pick a one-sided non-orienting loop such that the scheme $G'$ that we obtain when removing it and flipping its wedge maximizes the number of one-sided loops. We recurse of $G'$ and apply the one-sided loop move on this loop.
    
    \item \textbf{Step 3.2.a: If all one-sided loops are orienting and there are two-sided loops.} We pick an orienting loop adjacent to a two-sided loop, recurse on the drawing $H'$ described in the concatenation move and apply the concatenation move on these loops.
    
    \item \textbf{Step 3.2.b: If all one-sided loops are orienting and there are no two-sided loops.} In this case one cross-cap is sufficient to draw all the loops.
\end{itemize}
\textbf{Post-processing steps:}
\begin{itemize}
    \item \textbf{Step B'.} Erase the extra separating loops added in step $B$.
    \item \textbf{Step A'.} Uncontract the spanning tree and remove the orienting loop added in step A. 
\end{itemize}
\end{tcolorbox}

The numbering comes from the Schaefer and {\v{S}}tefankovi{\v{c}} algorithm. A main difference with our modified version is that it is not clear at first sight that we cover all cases. This follows from the presence of an orienting loop and will be the main purpose of the following subsection(s). After that we will be ready to prove in Lemma~\ref{modified algorithm} that our algorithm terminates and outputs a cross-cap drawing where a loop does not enter each cross-cap more than $6$ times.

\subsubsection{Completeness of the case analysis }

The following lemma explains why we do not need to consider the case in which all loops are two-sided.

\begin{lemma}\label{non-orientability}
A scheme with an orienting loop is non-orientable.
\end{lemma}

\begin{proof}
Let us assume that $G$ is an orientable scheme, hence the orienting loop $o$ is two-sided. By Lemma \ref{odd non-orientable genus}, $G$ needs $eg(G)+1$ cross-caps to embed and we know that $eg(G)$ is twice the orientable genus of $G$. Therefore $G$ needs an odd number of cross-caps to embed and the loop $o$ passes through each of them an odd number of times. 
Therefore it is one-sided which is a contradiction.
\end{proof}

Lemma~\ref{non-orientability} guarantees that we do not need to consider the case where all loops are two-sided when the algorithm starts, but this case might a priori still happen during recursive calls to the algorithms. Fortunately, this will actually not be the case, as we will prove in Corollary~\ref{conservation of orienting } that there is always an orienting loop in each of the recursive calls.\\

The following lemma explains why there is only one case that can happen in step 2 of our algorithm.
\begin{lemma}\label{orienting+separating}
Let $G$ be a scheme with an orienting loop $o$ and a non-contractible separating loop $s$. Then $s$ separates the graph into an orientable and a non-orientable sub-graph.
\end{lemma}

\begin{proof}
By Lemma \ref{non-orientability}, $G$ is non-orientable and therefore it has at least one one-sided loop. Let us assume that $s$ separates the scheme into $G_1$ which inherits the loop $o$, and $G_2$. We show that $G_1$ is non-orientable and $G_2$ is orientable.

By Lemma~\ref{orienting loops}, any one-sided loop has exactly one end in the wedge of the orienting loop in $G_1$ and has to have both ends in the same side of the separating loop, therefore no one-sided loop can exist in $G_2$ so $G_1$ is non-orientable and $G_2$ is orientable. The case where the only one-sided loop is the orienting loop is trivial.

\end{proof}

 \subsubsection{The order on the one-sided non-orienting loops}
 
 In this section, we explain an order on one-sided loops when we apply the one-sided loop move in step 3.1 of the algorithm. The reason for imposing this restriction is to avoid creating separating loops in the induction that did not exist in the scheme initially.

\begin{lemma}\label{one nonorientable component}
If there exists an orienting loop $o$ (two-sided or one-sided) in the embedding scheme $G$, the connected component that has the vertex $o$ is the only non-orientable component in $I_G$.\end{lemma}
\begin{proof}

By Lemma \ref{orienting loops}, the ends of every one-sided loop interleave with the ends of the orienting loop, therefore the vertex $o$ is connected to every one-sided vertex in $I_G$ and this finishes the proof; see Figure  \ref{one-sided moves}.\end{proof}

The following lemma is analogous to a similar result in signed reversal distance theory~\cite[Fact~2]{bergeron2001very}.

\begin{lemma}\label{one-sided move in interleaving }
Applying a one-sided loop move on a loop $r$ corresponds to removing the vertex $r$ in the interleaving graph and complementing the subgraph induced by its neighbors.
\end{lemma}

\begin{proof}
In a one-sided loop move, we remove a one-sided loop $r$ and flip one of its wedges. When we flip, we change the signature of each loop that has exactly one end in each wedge of $r$, therefore we are changing the sidedness of everything that was connected to the vertex $r$ in the interleaving graph. Also, due to the flip, every two vertices adjacent to $r$, that were connected to each other before the flip, are now disconnected and vice versa. The situation of the loops that are not interleaving with $r$ in the scheme is unchanged; see Figure \ref{one-sided moves}. 
\end{proof}

 \begin{figure}[t]
    \centering
    \includegraphics[width=\textwidth]{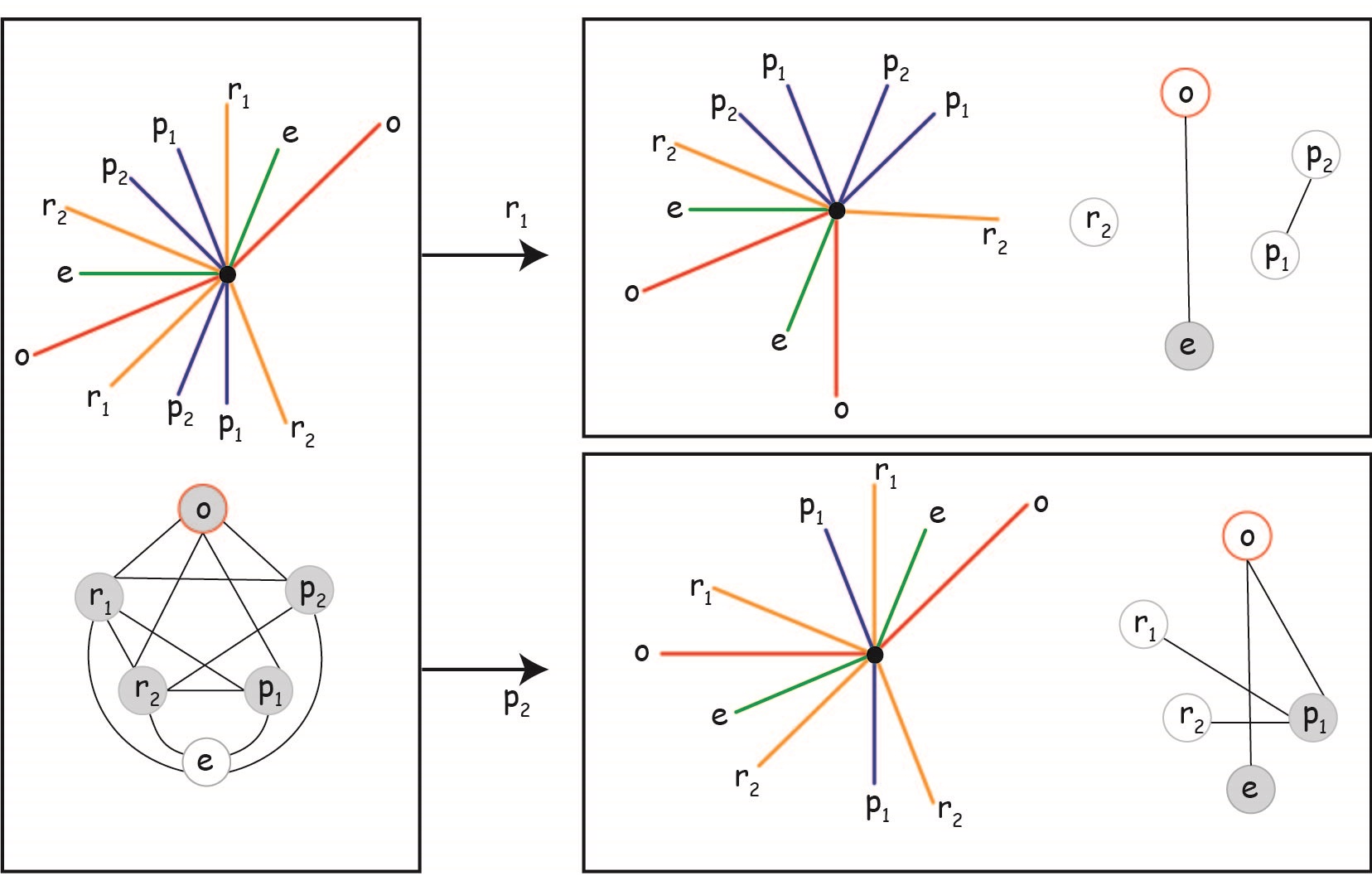}
     \caption{Each box represents a scheme with its interleaving graph, note that in all these schemes the loop $o$ is orienting; right: the impact of applying the one-sided loop on $r_1$ (top) and $p_2$ (bottom).  }
    \label{one-sided moves}
\end{figure}

Note that such a move can change the number of connected components in the interleaving graph, and might increase the number of orientable components, as is the case when we apply the one-sided loop move on $o$ in Figure~\ref{one-sided moves}.

\begin{lemma}\label{orienting in one-sided move }
Let $G$ be a scheme with orienting loop $o$ and a one-sided non-orienting loop $r$. Let $G^{'}$ be the graph we obtain after removing $r$ and flipping its wedge. The loop $o$ is orienting in $G^{'}$.
\end{lemma}

\begin{proof}
We need to show that after removing $r$, the loop $o$ is connected to all the one-sided vertices in the interleaving graph and it is not connected to any two-sided vertex (Lemma \ref{orienting loops}). We know that at the start $o$ is connected to $r$. Any one-sided loop that $r$ used to be adjacent to is now two-sided and since $o$ used to be connected to these loops, by complementing the subgraph induced by the adjacent vertices to $r$ (Lemma \ref{one-sided move in interleaving }), now it is not connected to them. Similarly, we can see that if $r$ used to be connected to two-sided loops, after flipping the wedge, they get one-sided and since the loop $o$ was not connected to any two-sided vertex before, now it gets connected to them. The situation for the loops to which $o$ was connected and $r$ was not, remain unchanged.
\end{proof}

\begin{lemma}\label{forienting}
Let $G$ be a scheme with only orienting and non-separating two sided loops. There exists an orienting loop $o$ followed by a non-separating two sided loop $t$. If we let $G^{'}$ be the graph that we obtain by replacing $o$ with the loop $o^{'}$ that is the concatenation of $o$ and $t$ in $G$, and $G^{''}$ be the graph we obtain by applying a one-sided loop move on $o^{'}$ in $G'$, then the loop $t$ is orienting in $G^{''}$.
\end{lemma}

\begin{proof}
Since there exist only orienting and two-sided non separating loops, and there is at least one of each, a pair of consecutive ones must exist.
First we show that $t$ and $o$ belong to different components in $I_G$. Since $o$ is orienting and there is no one-sided non-orienting loop in the scheme, then the component of $o$ is a complete graph with only orienting loops. Therefore, $t$ does not belong to this component and everything in the component of $t$ is two-sided. Replacing $o$ by $o^{'}$ corresponds to replacing the vertex $o^{'}$ with $o$ in $I_G$ and connecting it to the neighbors of $t$, since $o^{'}$ now interleaves with every loop that $t$ interleaves with.
Now applying the one-sided loop move on $o^{'}$ makes every neighbor of $t$ one sided and all the orienting vertices (that were formerly adjacent to $o$) two-sided. Therefore the only one-sided loops in the new scheme are the neighbors of $t$ and $t$ is not adjacent to any two-sided vertices since everything in its component was two-sided before. By Lemma \ref{orienting loops}, $t$ is orienting in $G^{''}$.
\end{proof}

Since the contractible loop move clearly preserves orienting loops, Lemmas~\ref{orienting+separating} (note that the dragging move first adds an orienting loop to the orientable part), \ref{orienting in one-sided move } and \ref{forienting} imply the following Corollary \ref{conservation of orienting }.
\begin{corollary}\label{conservation of orienting }
Let $G$ be a one-vertex scheme with an orienting loop. Let $G'$ be the graph on which the modified algorithm recurses when applying one of the following moves:
\begin{itemize}
   \item A contractible loop move.
   \item A one-sided loop move on a one-sided non-orienting loop.
    \item The concatenation move on an orienting loop.

\end{itemize}
Then $G'$ has an orienting loop. Likewise, when the modified algorithm applies a gluing and dragging move on a separating loop $s$, the two subgraphs $G_1$ and $G_2$ on which it recurses have an orienting loop.
\end{corollary}

The next lemma explains our choice of rule in Step 3.1:

 \begin{lemma}\label{first order}
Let $G$ be a one-vertex scheme with an orienting loop and no non-contractible separating loop such that $I_G$ has only one non-trivial component. Then $G$ can be drawn by exclusively applying a sequence of contractible loop moves, one-sided loop moves and concatenation moves.
 \end{lemma}

This ensures that in Step 3.1 no non-contractible separating loop is created during the process, hence we can avoid increasing the number of orientable components. This proof mirrors results in the signed reversal distance theory (see Bergeron~\cite[Theorem~1]{bergeron2001very}) in which similar claims are proved in the context of applying reversals on permutations. 

\begin{proof}
In order to have a non-contractible separating loop, it is necessary to have either two non-trivial components, or an orientable non-trivial component and a trivial non-orientable component (an isolated one-sided vertex). 

Since there exists an orienting loop in the scheme and the ends of the orienting loop interleaves with those of every one-sided loop, then all one-sided loops belong to the same component in the interleaving graph and thus there exists only one non-orientable component.
By Corollary \ref{conservation of orienting }, in applying these moves, there is always an orienting loop in every step. This, together with Lemma \ref{one nonorientable component}, implies that the number of non-orientable components remains 1 through the algorithm. Therefore it suffices to prove that we can draw $G$ such that in each step we are not increasing the number of non-trivial orientable components.

\textbf{The non-trivial component is non-orientable.}
If there exists no one-sided non-orienting loop, then every loop in the scheme is orienting or contractible. In this case, the non-orientable genus is one and the result is trivial. Let us assume that there exists at least a one-sided non-orienting loop. We claim that if we choose the one-sided non-orienting loop such that flipping its wedge maximizes the number of one-sided loops, then we do not increase the number of orientable components. In Figure \ref{one-sided moves}, it is shown that applying a one-sided loop move on $p_2$ maximizes the number of one-sided loops but applying it on $r_1$, increases the number of orientable components and also turns $r_2$ to a separating loop for the new scheme. 

Let us assume that applying a one-sided loop move on the one-sided loop $r$ maximizes the number of one-sided loops and increases the number of non-trivial orientable components.
Let $i$ be a vertex that was adjacent to $r$, belonging to a component $U$ that got disconnected when complementing the subgraph induced by the neighbors of $r$. The vertex $i$ was one-sided and we claim that taking $i$ instead of $r$ would have created more one-sided vertices which is a contradiction.\\
Denote by $O(r)$ ($N(r)$) the number of two-sided (one-sided) loops adjacent to $r$.
Removing a one-sided loop is equivalent to removing the vertex in the interleaving graph and complementing the subgraph induced by its adjacent vertices. Therefore an edge $r$ increases the number of one-sided loops in the scheme by $O(r)-N(r)$.
All two-sided vertices adjacent to $e$ are one-sided after removing and therefore they are not in $U$, meaning that they were formerly connected to $i$. Therefore $O(i)\geq O(r)$.\\
Similarly, $r$ has to be connected to every one-sided vertices that were formerly connected to $i$, therefore $N(r)\geq N(i)$.\\
If $O(i)= O(r)$ and $N(r)=N(i)$, it means that they have the same neighbors and removing $r$ will isolate $i$ which contradicts the fact that the connected component $U$ is not trivial.
Therefore, applying a one-sided loop move on the loop $i$ creates more one-sided loops than the loop $r$, which is a contradiction. Thus, removing $r$ cannot add to the number of non-trivial orientable components.

\textbf{The non-trivial component is orientable.}
In this case, there exists only one one-sided loop $o$ which is orienting. Then we replace $o$ with its concatenation with one of the two-sided loops that is immediately next to it in the rotation. Denote this two-sided loop by $t$ and the concatenated loop by $o^{'}$. This corresponds to replacing the vertex $o$ by $o^{'}$ that is connected to all the neighbors of $t$ in the interleaving graph and therefore reduces the number of components by 1. Applying the one-sided loop move on $o^{'}$, isolates $t$ and makes it orienting for the new scheme (Lemma \ref{forienting}). The resulting graph falls in the last case where the non-trivial component is non-orientable and therefore we can draw it by applying only contractible loop moves and one-sided loop moves. 
This completes the proof.\end{proof}

\subsubsection{Correctness of the modified algorithm}

\begin{lemma}\label{modified algorithm}
Let $G$ be a graph cellularly embedded on a non-orientable surface. If $G$ has an orienting loop, applying the modified algorithm, we obtain a cross-cap drawing of $G$ with $eg(G)$ cross-caps such that each loop of $G$ enters each cross-cap at most twice. Otherwise, we obtain a cross-cap drawing of $G$ with $eg(G)$ cross-caps such that each loop of $G$ enters each cross-cap at most 6 times.
\end{lemma}

\begin{proof}
By Lemma \ref{reducing}, Step A in the algorithm reduces the graph $G$ to a one-vertex scheme $\hat{G}$ that has an orienting loop such that a drawing $\hat{G}$ leads to a drawing for $G$. Let $\bar{G}$ be the scheme that we obtain after step B on $\hat{G}$. This step does not interfere with the Euler genus of the scheme. 

Thus, by Remark~\ref{coefficient3}, it is sufficient to prove that there is a cross-cap drawing for $\bar{G}$ with $eg(G)=eg(\bar{G})$ cross-caps in which each edge is passing through each cross-cap at most twice. We prove this claim for any one-vertex scheme $G$ with an orienting loop.

In order to prove this claim, we follow the recursive steps of the main loop of the modified algorithm, and thus provide a proof by induction on $eg(G)+|E(G)|$. \\

\textbf{Step 1.} We apply the contractible loop move on every contractible loop. By the induction hypothesis, we can obtain a drawing with $eg(G)$ cross-caps for the resulting scheme in which every loop passes through each cross-cap at most two times and the contractible loop is not using any of them.

\textbf{Step 2.} If there exists separating (non-contractible) loops, we deal with them in the prescribed order.
Take the separating loop $s$ and divide the scheme. By Lemma \ref{orienting+separating}, we know that one of these sub-graphs is orientable and the other one is non-orientable, without loss of generality let us assume that $G_1$ is non-orientable. We apply a combination of the gluing move and the dragging move on these sub-graphs: we add an auxiliary orienting loop $o$ to $G_2$. By the induction hypothesis, there are cross-cap drawings $H_1$ with $eg(G_1)$ cross-caps and a drawing $H_2$ with $eg(G_2)+1$ cross-caps for $G_1$ and $G_2+\{o\}$ so that each edge of $G_1$ and $G_2$ passes through each cross-cap at most twice. Let $H_{2}^{'}$ be the drawing for $G_2+\{o\}+\{s\}$ that we obtain as described in the dragging move. By Remark \ref{enteringonce} and by the induction hypothesis, we know that in $H_2$, the loop $o$ passes through each cross-cap exactly once. The loop $s$ is following $o$ twice and therefore it passes through each cross-cap in $H_2$ exactly twice.

The gluing move does not interact with the number of entrances for any loop. In the dragging move, every loop that is being dragged from $H_1$ to $H_2$ is following the auxiliary orienting loop $o$ in $G_2$.  Since each edge in $H_1$ passes through each cross-cap at most twice, therefore at most two of its sub-edges are being dragged along $o$ and therefore they pass through the cross-caps in $H_2$ at most twice.

Since our graph is non-orientable, and till here we only dealt with two-sided loops, not all of the edges can be orientable at this point. 

\textbf{Step 3.1.}
If there is one-sided non-orienting loop $r$ in $G$, we apply a one-sided loop move on the one that respects our prescribed order. Since $r$ is non-orienting, the scheme $G^{'}$ obtained after removing $r$ and flipping its wedge is still non-orientable and by the induction hypothesis, there is a drawing $H^{'}$ with $eg(G)-1$ cross-caps for $G^{'}$ in which every loop passes through each cross-cap at most twice. After drawing $r$ and the new cross-cap, we can see that each loop in the flipped wedge passes through this cross-cap at most twice (the exact value depending on the number of its ends that it has in the flipped wedge).

\textbf{Step 3.2.a} We can find an orienting loop $o$ that is followed immediately by a two-sided loop $t$. We apply the concatenation move on $o$ and replace it with the non-orienting loop $o^{'}$ (denote by $G^{'}$, the scheme we obtain at this point; note that $eg(G^{'})=eg(G)$) and then we apply the one-sided loop move on $o^{'}$. By the induction hypothesis, there is a drawing $H^{'}$ for $G^{'}$ with $eg(G)$ cross-caps in which each loop passes through each cross-cap at most twice. Since we first apply a one-sided loop move on $o^{'}$, we can see that $o^{'}$ uses only one cross-cap. Depending on our choice in flipping a wedge of $o^{'}$, the loop $t$ is using this cross-cap either twice or not at all (the loop $t$ uses every cross-cap except this cross-cap exactly once, since after removing $o^{'}$, $t$ gets orienting, see Lemma \ref{forienting}). One can check that at the end for both choices of wedge, the loop $o$ passes through each cross-cap exactly once.

\textbf{Step 3.2.b} If all one-sided loops in the scheme are orienting and there is no two-sided loop in the scheme, then all of the orienting loops in the scheme are homotopic and we can draw them using one cross-cap with all of the loops passing though the cross-cap exactly once.\\

By Corollary \ref{conservation of orienting }, the graph has an orienting loop at each step of the algorithm and therefore by Lemma \ref{orienting+separating}, we never have a graph in which every loop is two-sided throughout the algorithm.
This completes the proof of the claim and we conclude. 
\end{proof}

Note that this proof is independent of the orders we defined for the loops in Step 2 and Step 3.1. These orders are shown to be useful in the next section.

Except the fact that our proof do not cover all the steps that happen in the original case (the case that there might not exist an orienting loop in the scheme), another difference between this proof and the proof of the Schaefer and {\v{S}}tefankovi{\v{c}} algorithm is in Step 3.2. In this Step, the original algorithm flips the wedge of $o^{'}$ that does not encompass the loop $t$. We prove the step for both choices of wedge because we favour the freedom to choose a wedge that we want to flip for our further purposes.

\subsection{The Non-orientable Canonical System of Loops}\label{the system}

The modified algorithm that we described in the previous sections provides us with a cross-cap drawing of any embedded graph $G$ where each edge of the graph enters each cross-cap at most six times, as per Lemma~\ref{modified algorithm}. Furthermore, our algorithm has the following key advantage compared to the algorithm of Schaefer and \v{S}tefankovi\v{c}: due to the order in which we choose the loops in Steps 2 and 3.1, we know that dragging moves and the other moves do not intermingle during the recursive calls of the algorithm. Indeed, first, by Lemma~\ref{first order}, when it draws a scheme with a single non-trivial component, it only relies on contractible, one-sided and concatenation moves. Second, due to the order in which we choose the loops in Steps 2 and 3.1, we know that whenever a dragging move is applied, the orientable sub-scheme on which we recurse has only one non-trivial component. In this section, we leverage these two key advantages to find a non-orientable canonical system of loops of small multiplicity.

\subsubsection{The Dual Graph of The Cross-cap Drawing}
In this rather tedious but straightforward section, we first investigate the effect of every move involved in the modified algorithm on the dual graph of the cross-cap drawing (viewed as a planar graph). Every edge  $e$ in a cross-cap drawing $H$, corresponds to an edge  $e^{*}$ and every face $f$ corresponds to a vertex $f^{*}$ in the dual graph. The vertex $v$ corresponds to the face $v^{*}$ and the cross-caps correspond to the other faces in the dual graph.

By Remark \ref{enteringonce}, in the dual graph to the drawing obtained via the modified algorithm, the orienting loop passes through each cross-cap exactly once. Thus if $G$ is drawn using $k$ cross-caps, an orienting loop in $G$, corresponds to a set of $k+1$ edges in the dual graph. Furthermore, there are exactly 2 dual edges corresponding to the orienting loop in each face of the dual.

The effect of a contractible loop move is as follows:
\begin{lemma}\label{contractible move in dual}
Drawing a contractible loop in a face $f$ of the cross-cap drawing corresponds to adding a vertex with degree one attached to the vertex $f^{*}$.
\end{lemma}
\begin{proof}
See Figure \ref{concatenating in dual}.
This case does not change the situation of any of the dual edges corresponding to the orienting loop.
\end{proof}

Let us assume there is a loop $s$ that separates $G$ into $G_1$ and $G_2$, where $G_1$ is non-orientable and $G_2$ is orientable. We glue the drawings $H_1$ and $H_2$ for $G_1$ and $G_2$ and then we apply a dragging move to this case. The following lemma explains what this move does to the dual graph.
In the lemma we use the notation introduced in the description of the dragging move. We denote the vertex associated to the root face $f^{i}_o$ in $H_i$ by $f_{o}^{i*}$ for $i\in\{1,2\}$. We use the notation $(i_1,f_1,\dots,i_{2K}, f_{o}^1)$ for the sequence of edges and faces around the eliminated cross-cap in the dragging move, and finally we denote by $o$ the auxiliary orienting loop drawn in $H_2$.

\begin{lemma}\label{separating in dual}
 Let $s$ be a loop that separates the scheme $G$ into the non-orientable subgraph $G_1$ and the orientable subgraph $G_2$. In this case, the gluing move, the dragging move and later drawing back the separating loop corresponds to:
\begin{itemize}

\item splitting $f_{o}^{1*}$ into two vertices $f_{o}^{11*}$ and $f_{o}^{12*}$ such that $f_{o}^{11*}$ inherits only $i^{*}_{1}$ and $f_{o}^{12*}$ inherits the rest of the edges incident to $f_{o}^{1*}$,
    \item removing the two dual edges corresponding to the loop $o$ incident to the vertex $f_{o}^{2*}$ in $v_{2}^{*}$,
    \item connecting $f_{o}^{11*}$ and $f_{o}^{12*}$ to the adjacent vertices to $o$ in the correct order by adding an edge for each one (this edges correspond to the segments of the separating loop that are attached to the vertex),
    \item connecting $f_{o}^{2*}$ to $f_{k}^{*}$ by adding an edge,
    \item replacing the dual edges corresponding to the segments of $o$ in $H_2$ by $k+2$ edges.
\end{itemize}
These operations are pictured in Figure \ref{glue+drag in dual}.
\end{lemma}
\begin{proof}
Splitting $f_{o}^{1*}$ and connecting it to two vertices formerly incident to $f_{o}^{1*}$ corresponds to the gluing move between the drawings. We can see in Figure \ref{glue+drag in dual} that these steps merge the faces $v_{1}^{*}$ and $v_{2}^{*}$ to the face $v^{*}$. Note that the edges that we add to connect these vertices correspond to the sub-edges of the separating loop $s$ that are incident to the vertex.  On the other hand, connecting $f_{o}^{2*}$ to $f_{k}^{*}$ and replacing the dual edges in $H_2$ by a path corresponds to the dragging move; see Figure \ref{glue+drag in dual}.

 \begin{figure}[t]
    \centering
    \includegraphics[width=\textwidth]{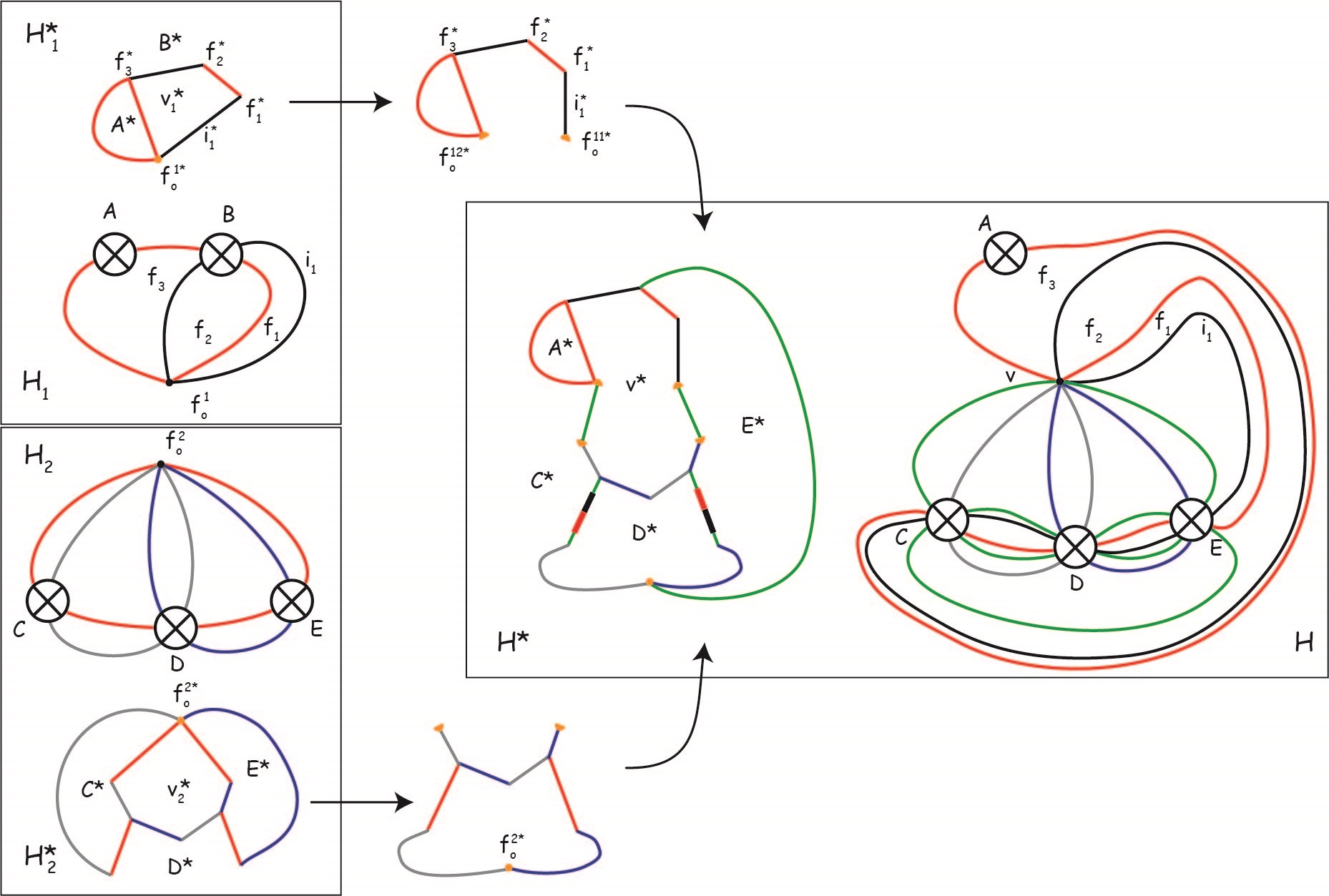}
     \caption{The impact of the gluing move and dragging move on the dual graph.}
    \label{glue+drag in dual}
\end{figure}

\end{proof}

\begin{remark}\label{former outrfaces}
Let us denote by $H$ the drawing of $G$ that we obtain by applying the five moves in Lemma \ref{separating in dual}. The only modification done on the sub-graph induced by the vertices that come from $H_{1}^{*}$ in $H^*$ is that we split the vertex $f_{o}^{1*}$ to  $f_{o}^{11*}$ and  $f_{o}^{12*}$. Note that both $f_{o}^{11}$ and $f_{o}^{12}$ are root faces in $G$.

The modifications in the subgraph induced by the vertices that come from $H^{*}_2$ is that we replace some edges by paths and we disconnect the two root faces adjacent to $f_{o}^{2*}$ by removing the incident edges. Also $f_{o}^{2*}$ gets connected to a face in $H_{1}^{*}$ and it is not necessarily a root face anymore.
\end{remark}

The effect of a one-sided loop move is as follows:
\begin{lemma}\label{1sided non-orienting in dual}
Adding a one-sided non-orienting loop $r$ with ends in root faces $a$ and $b$ in the drawing ($a$ and $b$ are possibly identical), corresponds to subdividing the face $v^{*}$ into two faces and adding a path of length $k+2$ from $a^{*}$ to $b^{*}$ where $k$ is the length of one of the paths from $a^{*}$ to $b^{*}$ in the face $v^{*}$. 
\end{lemma}

\begin{proof}
Figure \ref{one-sided move in dual} depicts that the addition of the new cross-cap in the one-sided loop move corresponds to adding a duplicate of the set of edges and vertices between $a^{*}$ and $b^{*}$ in the face $v^{*}$.  Choosing different wedges of $r$ corresponds to choosing between the two different sequences from $a^{*}$ to $b^{*}$ in the face $v^{*}$ to duplicate.
\end{proof}
 \begin{figure}[t]
    \centering
    \includegraphics[width=.75\textwidth]{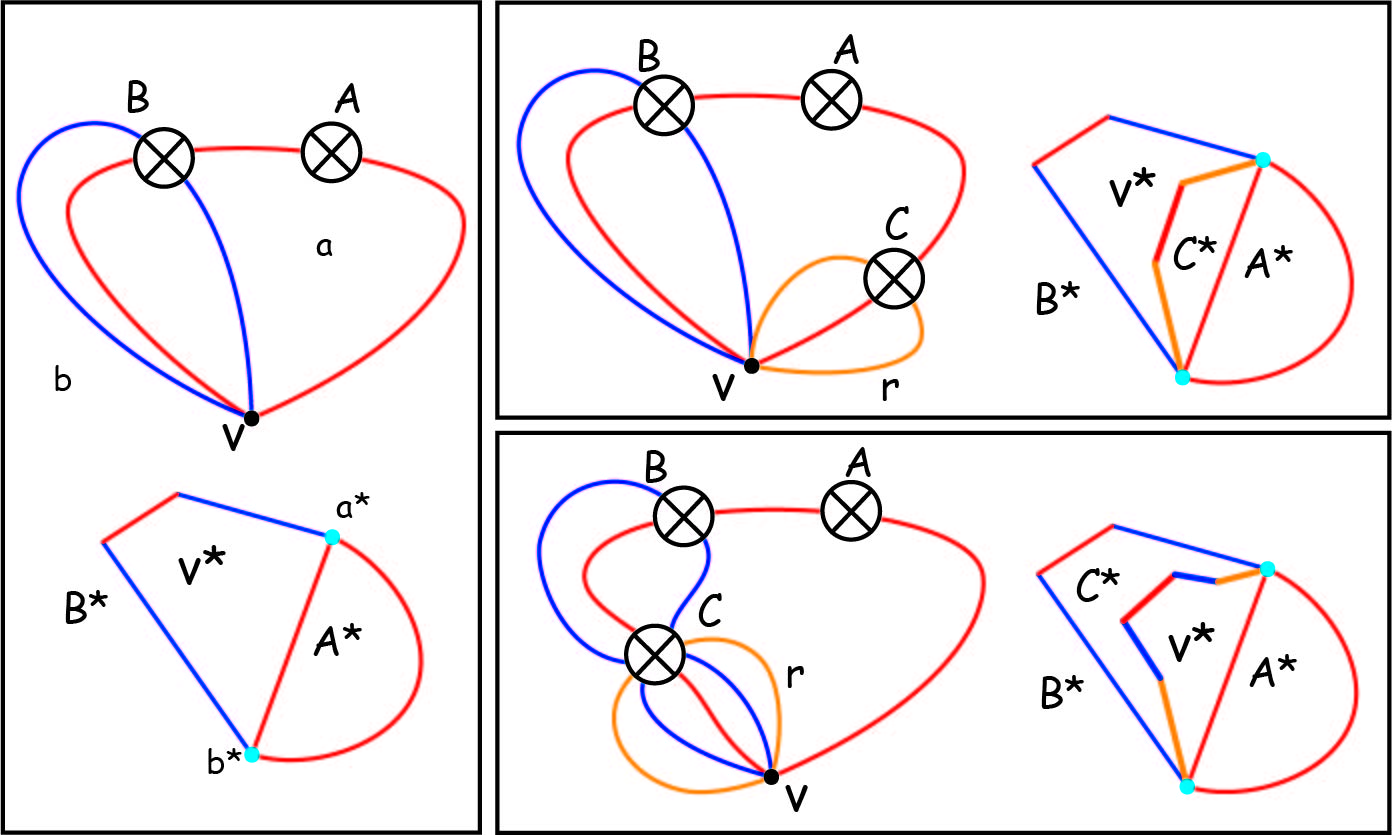}
     \caption{The impact of the one-sided loop move on the dual graph. Left: the drawing for the scheme before adding the one sided loop $r$ and its dual; right: the drawings and the dual graphs after drawing $r$ and the impact of flipping different wedges of $r$ is depicted.}
    \label{one-sided move in dual}
\end{figure}

Note that if there exists an orienting loop, it has exactly one end in each wedge of the loop $r$ and therefore $a^{*}$ and $b^{*}$ are separating the two edges corresponding to the orienting loop in the face $v^{*}$. The copy of this edge in the added path is dual to the new segment of the orienting loop. The yellow loop in Figure \ref{one-sided move in dual} is orienting and for example, the dual edge connecting the dual vertices $2^{'}$ and $3^{'}$ is the new dual edge associated to the orienting loop that is the copy of the edge connecting the dual vertex $2$ to $3$.

Finally, we analyze the effect of the concatenation move. Let $G$ be a scheme with no separating loop in which every one-sided loop is orienting and it has at least one two-sided loop. Let $o^{'}$ be the one-sided non-orienting loop obtained by concatenating the orienting loop $o$ and the two-sided loop $t$ which has an end immediately next to an end of $o$ in the rotation (step 3.2 of the algorithm). Denote by $G^{'}$ the scheme in which $o$ is replaced by $o^{'}$ and let $H^{'}$ be a drawing for $G^{'}$.  Note that after applying a one-sided loop move on $o^{'}$, $t$ gets orienting (by Lemma \ref{forienting}) and it goes through $eg(G)-1$ cross-caps. The loop $o^{'}$ is the last loop that is drawn in the algorithm and therefore it passes through only one cross-cap. Let us denote this cross-cap by $\mathfrak{c}$ . 
We denote by $t_i$, $1\leq i\leq eg(G)$, the dual edges to sub-edges of $t$ and by $o^{'}_1$ and $o^{'}_2$ the sub-edges of $o^{'}$, where $t_1$ corresponds to the sub-edge next to $o^{'}_1$.

Depending on the wedge of $o^{'}$ we choose to reverse, we proceed with concatenating $o^{'}$ with $t$ to get back the loop $o$ as discussed in the description of the concatenation move. The following lemma explains what happens in the dual graph of the cross-cap drawing when we apply the concatenation move in both cases.

The Lemmas \ref{1sided orienting in dual} and \ref{1sided orienting in dual(2)} explains the effect of the concatenation move in the dual graph:
\begin{lemma}\label{1sided orienting in dual}
When we reverse the wedge of $o^{'}$ that does not encompass $t$, concatenating the loop $o^{'}$ along $t$ corresponds to:
\begin{itemize}
    \item subdividing the edge adjacent to $t^{*}_{1}$ and $o^{'*}_1$ that is in the cyclic rotation of $\mathfrak{c}^{*}$
    \item contracting the edge $o^{'*}_1$
    \item subdividing the edges $t_{i}^{*}$ for $i\geq 2$
\end{itemize}
The added edges together with $o^{'*}_2$, correspond to the segments of $o$ in $H$.
\end{lemma}
\begin{proof}
As can be seen in Figure \ref{concatenating in dual}, by sliding $o^{'}$ along $t$, we remove $o^{'}_{1}$. Removing an edge corresponds to contracting its dual edge. Also it can be seen that following $t$ into the cross-caps adds parallel edges in the cross-cap drawing that corresponds to subdividing the dual edges $t_{i}^{*}$.
\end{proof}

 \begin{figure}[t]
    \centering
     \includegraphics[width=0.75\textwidth]{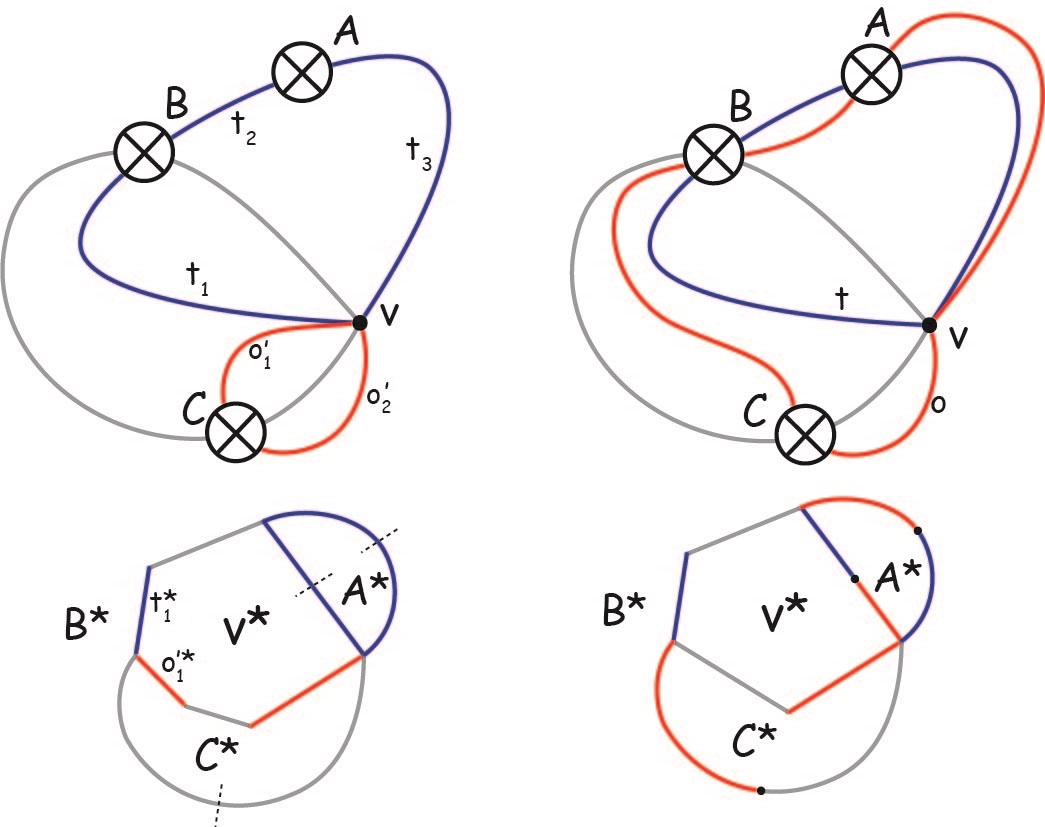}
     \caption{The impact of the concatenation move in the dual graph. The bottom graphs are dual  to the top cross-cap drawings.}
    \label{concatenating in dual}
\end{figure}

\begin{lemma}\label{1sided orienting in dual(2)}
When we reverse the wedge of $o^{'}$ that encompasses the ends of $t$, concatenating the loop $o^{'}$ with $t$ corresponds to:
\begin{itemize}
    \item subdividing every $t^{*}_i$ except for $i=1,2$
    \item subdividing the edge adjacent to $t^{*}_2$ and $o^{'*}_1$ in the face $v^{*}$ 
    \item contracting the edge $o^{'*}_1$ and $o^{'*}_2$
\end{itemize}
The added edges together with $o^{'*}_2$, correspond to the segments of $o$ in $H$.
\end{lemma}
\begin{proof}
The proof is similar to the proof of Lemma \ref{1sided orienting in dual}.
\end{proof}

\subsubsection{Short Paths from Each Cross-cap to The Vertex}

Recall that a root face in a cross-cap drawing of a one-vertex graph is a face of the drawing (seen as a planar graph) adjacent to the vertex. The aim of this section is to show that there is a cross-cap drawing output by the modified algorithm, in which the cross-caps are not too far from the vertex (at distance $O(|E(G)|)$). 
To show this, we find paths in the dual graph of the cross-cap drawing from a vertex adjacent to the face dual to each cross-cap to the vertex corresponding to a root face.

We first show this claim for a scheme with an orienting loop that has exactly one non-trivial component in its interleaving graph. Additionally, we claim that we can find a cross-cap drawing which allows us to force the paths to arrive in the same vertex in the dual graph that is chosen arbitrarily. Furthermore, we show that we can find these paths such that they do not use the dual edges corresponding to the orienting loop. To prove this we use the modified algorithm, but with an additional specification: as we mentioned before, different choices in flipping a wedge when doing a one-sided loop move or a concatenation move yield different cross-cap drawings. The modified algorithm that we described it gives us the freedom to choose the wedge whenever a one-sided loop move is applied. Here we use this freedom to build our desired paths. For ease of reading, we abuse language slightly and still refer to this more precise algorithm as the modified algorithm, the details of the choice of wedges being in the proofs of the following two lemmas.

\begin{lemma}\label{paths in one component graph}
For any one-vertex scheme $G$ with an orienting loop $o$ that has exactly one non-trivial component in its interleaving graph $I_G$, for any choice of root wedge $\omega$, the modified algorithm outputs a cross-cap drawing $H$ with $eg(G)$ cross-caps such that for every cross-cap $\mathfrak{c}$, there is a dual path $p_\mathfrak{c}$ from a face adjacent to $\mathfrak{c}$ to the root face corresponding to $\omega$ with multiplicity two, which does not cross the orienting loop. 
\end{lemma}

\begin{proof}
We prove the result by induction on $eg(G)+|E(G)|$, following the recursive steps of the modified algorithm. By Lemma \ref{first order}, we know that the modified algorithm draws such a scheme using only contractible loop moves, one-sided loop moves and concatenation moves. In this proof, we show that these moves can be applied in each step such that they do not increase the multiplicity of the paths that we obtain by the induction hypothesis. Crucially, when applying a one-sided loop move or a concatenation move, this relies on choosing a correct wedge to flip in each step. 

We fix an arbitrary root wedge $\omega$ around the vertex $v$. When we remove a loop that affects $\omega$, we update $\omega$ to be the wedge that encompasses the former fixed root wedge $\omega$. Similarly, when we re-introduce the edges in the drawing, we subdivide $\omega$ and we choose the sub-wedge that is consistent with the first choice of $\omega$.

\textbf{Contractible loop move.}
Denote by $G^{'}$ the scheme we have after removing a contractible loop. By the induction hypothesis, there exists a drawing $H^{'}$ and a system of paths $\{p_\mathfrak{c}\}$ from every cross-cap to the wedge $\omega$ in $H^{'*}$ with multiplicity two such that they do no use the dual edges of the orienting loop. When re-introducing $c$, if $c$ does not sub-divide the wedge $\omega$, then it can be seen that it does not affect the paths $p_\mathfrak{c}$. In the case that we need to update $\omega$ to be the empty wedge between the ends of $c$ itself, by Lemma \ref{contractible move in dual}, we can see that in this case, the paths need to use the edge $c^*$ in the dual once and therefore the multiplicity remains two. Also one can see that the paths still do not use the dual edges of the orienting loop.

\textbf{One sided loop move on a one-sided non-orienting loop $r$.}
Denote by $G^{'}$ the graph we have after applying a one-sided loop move on $r$. By the induction hypothesis, there exists a drawing $H^{'}$ and a system of paths $\{p_\mathfrak{c}\}$ from every cross-cap to the wedge $\omega$ in $H^{'*}$ with multiplicity at most two such that they do not use the dual edges of the orienting loop. In this case, we flip the wedge of $r$ that does not contain $\omega$. By this choice, we can see that the situation of the vertex dual to the face $\omega$ does not change, since by Lemma \ref{1sided non-orienting in dual}, drawing a one-sided loop corresponds to adding a path to the dual graph that does not separate $\omega$ from $v^*$. Therefore, this move does not affect any $p_\mathfrak{c}$.

We know that the orienting loop interleaves with $r$. For the new cross-cap $\mathfrak{c}_1$, we define the path $p_{\mathfrak{c}_1}$ as follows. We choose a face adjacent to $\mathfrak{c}$ that is a root face, and such that it is both in the flipped wedge of $r$ and in the same wedge of the orienting loop as $\omega$; see Figure \ref{newpath}. We introduce the path $p_{\mathfrak{c}_1}$ to be the sequence of dual edges and vertices around $v$ between $\omega$ and this root face such that it does not use the dual edges corresponding to the orienting loop. Since each edge has exactly two ends around the vertex, the path $p_{\mathfrak{c}_1}$ uses each edge at most twice in the dual graph and its multiplicity is at most two.
By construction, all these paths avoid using the dual edges of the orienting loop.

\begin{figure}[t]
    \centering
    \includegraphics[width=0.35\textwidth]{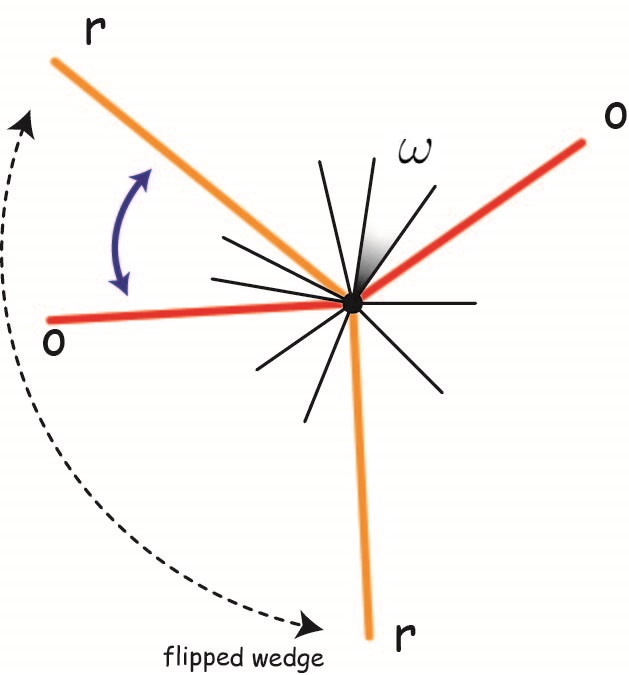}
     \caption{The choice of the adjacent root face for the new cross-cap: the loop $o$ is orienting and we applied the one-sided loop move on the loop $r$. The dotted arrow shows the wedge of $r$ that we flipped and the blue arrow shows the adjacent root wedges to the last cross-cap that we can choose from.}
   \label{newpath}
\end{figure}

 \textbf{Concatenation move on an orienting loop $o$ and the two sided loop $t$.}
 We denote by $o^{'}$ the concatenation of $o$ and $t$ and by $G^{'}$ the graph in which we replaced $o$ by $o^{'}$. The modified algorithm applies the one-sided loop move on $o^{'}$, and here again we choose to flip the wedge of $o^{'}$ that does not encompass $\omega$. We denote the new graph after removing $o^{'}$ by $G^{''}$. By the induction hypothesis there exists a drawing $H^{''}$ for $G^{''}$ and a suitable system of dual paths $\{p_{\mathfrak{c}}\}$ in $H^{''*}$ from a face adjacent to each cross-cap to the root wedge $\omega$ with multiplicity two. These paths do not use the dual edges corresponding to the loop $t$ since after removing $o^{'}$, $t$ is orienting for the new scheme (this follows from Lemma \ref{forienting}). Similar to the case before in which the algorithm applies the one-sided loop move, we can see that after drawing $o^{'}$ and adding a cross-cap, we can re-introduce the paths $\{p_{\mathfrak{c}}\}$ with the same multiplicity so that they do not use the two dual edges corresponding to $o^{'}$.  Now, the modified algorithm slides $o^{'}$ along $t$ to get a drawing for the initial scheme. By Lemmas \ref{1sided orienting in dual} and \ref{1sided orienting in dual(2)} (depending on the situation of $\omega$ with respect to $t$ and $o^{'}$), we know that sliding $o^{'}$ along $t$ corresponds to sub-dividing the dual edges corresponding to $t$ and since the paths $\{p_i\}$ do not use the dual edges of $t$, then they do not use the dual edges of $o$ either. For the last added cross-cap, we take a root face adjacent to it in the same wedge that $\omega$ is placed and introduce a path by going around the vertex. As before, we know that this path has multiplicity at most two and does not use the dual edges of the orienting loop. This finishes the proof.\end{proof}

Using Lemma \ref{paths in one component graph}, we prove the claim for the more general case in which the scheme can have more than one non-trivial component and we do not need the paths to arrive in the same root face.

\begin{lemma}\label{paths in graphs}

For any saturated one-vertex scheme $G$ with an orienting loop $o$, the cross-cap drawing $H$ output by the modified algorithm has $eg(G)$ cross-caps, and there is a path from every cross-cap to a root face (not necessarily fixed) with multiplicity at most two.

\end{lemma}

\begin{proof}

The proof is by induction on the number of separating loops.
When there is no separating loop, the graph has only one non-trivial component and it is non-orientable. In this case, the result follows by Lemma~\ref{paths in one component graph}.

 Let $s_l$ be the separating loop chosen by the algorithm during Step 2, separating two sub-scheme $G_l$ and $G \setminus G_l$ on which it recurses. Since $s_l$ separates a leaf from the component tree, one of these sub-schemes, say $G\setminus G_l$, is non-orientable and has an orienting loop. Therefore, by the induction hypothesis, there is a drawing $H^{'}$ for $G\setminus G_l$ with $eg(G\setminus G_l)$ cross-caps such that there is a path with multiplicity two from every cross-cap to a root wedge. 
 
Now, $G_l$ is made of exactly one non-trivial component due to our way of choosing $s_l$. Let $\omega$ be a root wedge of $G_l$ different from $f_o$, the face where the ends of $s_l$ used to exist. We apply Lemma~\ref{paths in one component graph} to obtain a cross-cap drawing $H_l$ of $G_l+\{o\}$ and a system of dual paths $\{p_\mathfrak{c}\}$ with multiplicity at most two from a face adjacent to every cross-cap to $\omega$, such that none of them use the dual edges corresponding to the orienting loop $o$. Now, the algorithm glues $H_l$ to $H^{'}$ and proceeds with dragging the loops from $H^{'}$ to $H_l$. We denote the resulting drawing by $H$.

 By Remark \ref{former outrfaces}, we know that the paths connecting cross-caps to root wedges in $H^{'}$ can be re-introduced in $H$, since dual edges and vertices corresponding to the edges and faces in $H^{'}$ are not changed in $H$ except the vertex that is split into two vertices. Since both of these vertices are root faces in the new scheme, this does not interfere with the multiplicity of these paths and each of them arrives to one of these vertices (recall that we do not require all the paths to arrive at the same root wedge). By the choice of $\omega$ and the fact that none of the paths in $\{p_\mathfrak{c}\}$ use the dual edges corresponding to the orienting loop, none of the paths visit the vertex $f_{o}^{*}$ (note that $f_{o}$ and $\omega$ are in different wedges of the orienting loop $o$). Also since the paths $\{p_\mathfrak{c}\}$ do not use $o$, we can choose the incident face  to each cross-cap so that replacing the dual edges of $o$ by a sequence of edges (as explained in Lemma \ref{separating in dual}), does not impact the multiplicity of the paths $\{p_\mathfrak{c}\}$ from each cross-cap to $\omega$. This finishes the proof.\end{proof}

\subparagraph*{The proof of the main theorem}

\canonical*
\begin{proof}
Applying the algorithm on $G$, we obtain the saturated one-vertex scheme $\bar{G}$ that has an orienting loop after the preprocessing steps.
By Lemma \ref{reducing}, to prove the theorem, it is sufficient to show that there exists a canonical system of loops for a drawing of $\bar{G}$ such that each loop in the system has multiplicity 10.

The one-vertex scheme $\bar{G}$ has an orienting loop and therefore by Lemma \ref{paths in graphs}, there exists a cross-cap drawing $\bar{H}$ with $eg(N)$ cross-caps such that there are paths $\{p_j\}$ with multiplicity two from a face incident to each cross-cap (denote this face by $b_j$ for each $j$) to a root face (denote this face by $a_j$ for each $j$) in this cross-cap drawing.

Fix a root face $f$ in the drawing $\bar{H}$. For each path $p_j$, build a loop $\nu_j$ by going from $f$ to $a_j$, by going around the vertex in shortest way possible. By doing so, so far the loop has multiplicity at most two. Then follow $p_j$ to $b_j$: this adds at most two to the multiplicity since $p_j$ has multiplicity two. Then go into the cross-cap and come back to $b_j$ by going around it (this adds at most two to the multiplicity, since every edge passes through each cross-cap at most twice by Lemma~\ref{modified algorithm}) Finally, follow $p_j$ back to $a_j$ and then go back to $f$ from the same path (these two last steps add at most 4 to the multiplicity). Therefore each $\nu_j$ has multiplicity 10. By Lemma \ref{right combinatorics}, we know that the system of loops we obtain, is a non-orientable canonical system of loops. This finishes the proof.
\end{proof}

\bibliographystyle{plainurl}
\bibliography{bibliography.bib}

\end{document}